\documentclass[11pt]{article}
\usepackage[utf8]{inputenc}
\usepackage[boxed]{algorithm2e}
\usepackage{fullpage} %Include this to make your page wider
\usepackage{amssymb}
\usepackage[english]{babel}
\usepackage{amsthm}
\usepackage{braket}
\usepackage{amsmath}
\usepackage{tikz}
\usepackage{derivative}
\usepackage{tikz-qtree}

\usepackage{mathrsfs}
\usepackage{hyperref}
\hypersetup{
    colorlinks,
    linkcolor={blue!100!black},
    citecolor={blue!100!black},
}

\usepackage{graphicx}

\usepackage{mathtools}

\DeclareMathOperator{\cU}{\mathcal{U}}

\DeclareMathOperator{\cR}{\mathcal{R}}
\DeclareMathOperator{\cX}{\mathcal{X}}
\DeclareMathOperator{\cY}{\mathcal{Y}}
\DeclareMathOperator{\cH}{\mathcal{H}}
\DeclareMathOperator{\cZ}{\mathcal{Z}}

\newtheorem{theorem}{Theorem}[section]
\newtheorem{corollary}[theorem]{Corollary}
\newtheorem{lemma}[theorem]{Lemma}
\newtheorem{definition}[theorem]{Definition}
\newenvironment{claim}[1]{\par\noindent\underline{Claim:}\space#1}{}
\newenvironment{claimproof}[1]{\par\noindent\underline{Proof:}\space#1}{\hfill $\blacksquare$}

\newcommand{\E}{\mathbb{E}} 	%Expectation

\newcommand{\intg}{\mathbb{Z}}

\newcommand{\inProd}[2]{\langle{#1},{#2}\rangle}

\newcommand{\poly}{\text{poly}}

\newcommand{\DGauss}[2]{D_{#1\ifthenelse{\equal{#1}{}}{}{,#2}}}

\newcommand{\eps}{\varepsilon}
\renewcommand{\epsilon}{\varepsilon}

\newcommand{\norm}[1]{\left\lVert#1\right\rVert}

\title{Quantum Algorithms and Lower Bounds\\ for Linear Regression with Norm Constraints}
\author{Yanlin Chen\thanks{QuSoft, CWI, the Netherlands. {\tt yanlin.chen@cwi.nl}} 
\and 
Ronald de Wolf\thanks{QuSoft, CWI and University of Amsterdam, the Netherlands. Partially supported by the Dutch Research Council (NWO) through Gravitation-grant Quantum Software Consortium, 024.003.037, and through QuantERA ERA-NET Cofund project QuantAlgo 680-91-034. {\tt rdewolf@cwi.nl}} }
\date{}

\begin{document}

\maketitle

\begin{abstract}
Lasso and Ridge are important minimization problems in machine learning and statistics. They are versions of linear regression with squared loss where the vector $\theta\in\mathbb{R}^d$ of coefficients is constrained in either $\ell_1$-norm (for Lasso) or in $\ell_2$-norm (for Ridge).  We study the complexity of quantum algorithms for finding $\eps$-minimizers for these minimization problems.  We show that for Lasso we can get a quadratic quantum speedup in terms of $d$ by speeding up the cost-per-iteration of the Frank-Wolfe algorithm, while for Ridge the best quantum algorithms are linear in $d$, as are the best classical algorithms.  As a byproduct of our quantum lower bound for Lasso, we also prove the first classical lower bound for Lasso that is tight up to polylog-factors.
\end{abstract}

\section{Introduction}

\subsection{Linear regression with norm constraints}

One of the simplest, most useful and best-studied problems in machine learning and statistics is \emph{linear regression}. We are given $N$ data points $\{(x_i,y_i)\}_{i=0}^{N-1}$ where $x\in\mathbb{R}^d$ and $y\in\mathbb{R}$, and want to fit a line through these points that has small error. In other words, we want to find a vector $\theta\in\mathbb{R}^d$ of coefficients such that the inner product $\inProd{\theta}{x}=\sum_{j=1}^d\theta_jx_j$ is a good predictor for the $y$-variable.
There are different ways to quantify the error (``loss'') of such a $\theta$-vector, the most common being the squared error $(\inProd{\theta}{x}-y)^2$, averaged over the $N$ data points (or over an underlying distribution~$\mathcal{D}$ that generated the data).
If we let $X$ be the $N\times d$ matrix whose $N$ rows are the $x$-vectors of the data, then we want to find a $\theta\in\mathbb{R}^d$ that minimizes $\norm{X\theta-y}_2^2$.
This minimization problem has a well-known closed-form solution: $\theta=(X^T X)^+X^T y$, where the superscript `$+$' indicates the Moore-Penrose pseudoinverse.

In practice, unconstrained least-squares regression sometimes has problems with overfitting and often yields solutions $\theta$ where all entries are non-zero, even when only a few of the $d$ coordinates in the $x$-vector really matter and one would really hope for a sparse vector~$\theta$~\cite[see Chapters~2 and~13]{SB14}. This situation may be improved by ``regularizing'' $\theta$ via additional constraints. The most common constraints are to require that the $\ell_1$-norm or $\ell_2$-norm of $\theta$ is at most some bound $B$.\footnote{For ease of presentation we will set $B=1$. However, one can also set $B$ differently or even do a binary search over its values, finding a good $\theta$ for each of those values and selecting the best one at the end. Instead of putting a hard upper bound~$B$ on the norm, one may also include it as a penalty term in the objective function itself, by just minimizing the function $\norm{X\theta-y}_2^2+\lambda\norm{\theta}$, where $\lambda$ is a Lagrange multiplier and the norm of $\theta$ could be $\ell_1$ or $\ell_2$ (and could also be squared). This amounts to basically the same thing as our setup.}
Linear regression with an $\ell_1$-constraint is called \emph{Lasso} (due to Tibshirani~\cite{Tib96}), while with an $\ell_2$-constraint it is called \emph{Ridge} (due to Hoerl and Kennard~\cite{HK70}).

Both Lasso and Ridge are widely used for robust regression and sparse estimation in ML problems and elsewhere~\cite{Vin78,BG11}.
Consequently, there has been great interest in finding the fastest-possible algorithms for them. 
For reasons of efficiency, algorithms typically aim at finding not the exactly optimal solution but an $\eps$-minimizer, i.e., a vector $\theta$ whose loss is only an additive $\eps$ worse than the minimal-achievable loss. The best known results on the time complexity of classical algorithms for Lasso are an upper bound of $\tilde{\mathcal{O}}(d/\epsilon^2)$~\cite{HK11} and a lower bound of $\Omega(d/\epsilon)$~\cite{CSS11} (which we actually improve to a tight lower bound in this paper, see below);
for Ridge the best bound is $\tilde{\Theta}(d/\epsilon^2)$~\cite{HK11}, which is tight up to logarithmic factors.\footnote{For such bounds involving additive error $\eps$ to be meaningful, one has to put certain normalization assumptions on $X$ and $y$, which are given in the body of the paper. The $\mathcal{\tilde{O}}$ and $\tilde{\Theta}$-notation hides polylogarithmic factors. It is known that $N=\mathcal{O}((\log d)/\eps^2)$ data points suffice for finding an $\eps$-minimizer, which explains the absence of $N$ as a separate variable in these bounds.}

\subsection{Our results}
We focus on the \emph{quantum} complexity of Lasso and Ridge, investigating to what extent quantum algorithms can solve these problems faster.  Table~\ref{tab:linear_regression} summarizes the results.
The upper bounds are on time complexity (total number of elementary operations and queries to entries of the input vectors) while the lower bounds are on query complexity (which itself lower bounds time complexity).

\begin{table}[hbt]
\centering
\setlength{\tabcolsep}{20pt}
\renewcommand{\arraystretch}{1.5}
\begin{tabular}{|l| l| l| l|} 
 \hline
 \textbf{}  & \textbf{Upper bound} & \textbf{Lower bound} \\ [0.5ex] 
 \hline
\textbf{Lasso} & Classical~\cite{HK11}: $\tilde{\mathcal{O}}(d/\epsilon^2)$ & Classical [\emph{this work}]: $\tilde{\Omega}(d/\epsilon^2)$  \\
       &    Quantum [\emph{this work}]:   $\tilde{\mathcal{O}}(\sqrt{d}/\epsilon^2)$               &  Quantum [\emph{this work}]: $\Omega(\sqrt{d}/\epsilon^{1.5})$ \\
 \hline
 \textbf{Ridge} & Classical~\cite{HK11}: $\tilde{\mathcal{O}}(d/\epsilon^2)$& Classical~\cite{HK11}: $\Omega(d/\epsilon^2)$ \\
        &                    &  Quantum [\emph{this work}]: $\Omega(d/\epsilon)$\\
  \hline
\end{tabular}
\caption{Classical and quantum upper and lower bounds for Lasso and Ridge}
\label{tab:linear_regression}
\end{table}

\subsubsection{Lasso}
We design a quantum algorithm that finds an $\epsilon$-minimizer for Lasso in time  $\mathcal{\tilde{\mathcal O}}(\sqrt{d}/\epsilon^2)$. 
This gives a quadratic quantum speedup over the best-possible classical algorithm in terms of $d$, while the $\eps$-dependence remains the same as in the best known classical algorithm.

Our quantum algorithm is based on the Frank-Wolfe algorithm, a well-known iterative convex optimization method~\cite{FW56}. Frank-Wolfe, when applied to a Lasso instance, starts at the all-zero vector $\theta$ and updates this in $\mathcal{O}(1/\epsilon)$ iterations to find an $\epsilon$-minimizer.
Each iteration looks at the gradient of the loss function at the current point $\theta$ and selects the best among $2d$ directions for changing $\theta$ (each of the $d$ coordinates can change positively or negatively, whence $2d$ directions). The new $\theta$ will be a convex combination of the previous $\theta$ and this optimal direction of change. Note that Frank-Wolfe automatically generates \emph{sparse} solutions: only one coordinate of $\theta$ can change from zero to nonzero in one iteration, so the number of nonzero entries in the final $\theta$ is at most the number of iterations, which is $\mathcal{O}(1/\epsilon)$.

Our quantum version of Frank-Wolfe does not reduce the number of iterations, which remains $\mathcal{O}(1/\epsilon)$, but it does reduce the cost per iteration. In each iteration it selects the best among the $2d$ possible directions for changing $\theta$ by using a version of quantum minimum-finding on top of a quantum approximation algorithm for entries of the gradient (which in turn uses amplitude estimation). Both this minimum-finding and our approximation of entries of the gradient will result in approximation errors throughout.  Fortunately Frank-Wolfe is a very robust method which still converges if we carefully ensure those quantum-induced approximation errors are sufficiently small. 

Our quantum algorithm assumes coherent quantum query access to the entries of the data points $(x_i,y_i)$, as well as a relatively small QRAM (quantum-readable classical-writable classical memory). We use a variant of a QRAM data structure developed by Prakash and Kerenidis~\cite{Pra14,KP16}, to store the nonzero entries of our current solution $\theta$ in such a way that we can (1) quickly generate $\theta$ as a quantum state, and (2) quickly incorporate the change of $\theta$ incurred by a Frank-Wolfe  iteration.\footnote{Each iteration will actually change all nonzero entries of $\theta$ because the new $\theta$ is a convex combination of the old $\theta$ and a vector with one nonzero entry. Our data structure keeps track of a global scalar, which saves us the cost of separately adjusting all nonzero entries of $\theta$ in the data structure in each iteration.} Because our $\theta$ is $\mathcal{O}(1/\epsilon)$-sparse throughout the algorithm, we only need $\tilde{\mathcal{O}}(1/\epsilon)$ bits of QRAM.

We also prove a lower bound of $\Omega(\sqrt{d}/\epsilon^{1.5})$ quantum queries for Lasso, showing that the $d$-dependence of our quantum algorithm is essentially optimal, while our $\eps$-dependence might still be slightly improvable. Our lower bound strategy ``hides'' a subset of the columns of the data matrix~$X$ by letting those columns have slightly more $+1$s than $-1$, and observes that an approximate minimizer for Lasso allows us to recover this hidden set. 
We then use the composition property of the adversary lower bound~\cite{BL20} together with a worst-case to average-case reduction to obtain a quantum query lower bound for this hidden-set-finding problem, and hence for Lasso. 

Somewhat surprisingly, no tight \emph{classical} lower bound was known for Lasso prior to this work. To the best of our knowledge, the previous-best classical lower bound was  $\Omega(d/\epsilon)$, due to Cesa-Bianchi, Shalev-Shwartz, and Shamir~\cite{CSS11}.
As a byproduct of our quantum lower bound, we use the same set-hiding approach to prove for the first time the optimal (up to logarithmic factors) lower bound of $\tilde{\Omega}(d/\epsilon^2)$ queries for classical algorithms for Lasso (Appendix~\ref{app:classicalLBforLasso}).

\subsubsection{Ridge}
What about Ridge? Because $\ell_2$ is a more natural norm for quantum states than $\ell_1$, one might hope that Ridge is more amenable to quantum speedup than Lasso. 
Unfortunately this turns out to be wrong: we prove a quantum lower bound of $\Omega(d/\epsilon)$ queries for Ridge, using a similar strategy as for Lasso. This shows that the classical linear dependence of the runtime on $d$ cannot be improved on a quantum computer. Whether the $\eps$-dependence can be improved remains an open question.

\subsection{Related work}
As already cited in Table~\ref{tab:linear_regression}, Hazan and Koren~\cite{HK11} obtained an optimal classical algorithm for Ridge, and the best known classical algorithm for Lasso. Cesa-Bianchi, Shalev-Shwartz, and Shamir~\cite{CSS11} provided a non-optimal classical lower bound for Lasso, and their idea inspired us to hide a subset among the column of the data matrix and to use a Lasso solver to find that subset (our lower bound also benefited from the way composition of the adversary bound was used in~\cite{AW20}).

Du, Hsieh, Liu, You, and Tao~\cite{DHLYT20} also showed a quantum upper bound for Lasso based on quantizing parts of Frank-Wolfe, though their running time $\tilde{\mathcal{O}}(N^{3/2}\sqrt{d})$ is substantially worse than ours. The main goal of their paper was to establish differential privacy, not so much to obtain the best-possible quantum speedup for Lasso.
They also claim an $\Omega(\sqrt{d})$ lower bound for quantum algorithms for Lasso~\cite[Corollary~1]{DHLYT20}, without explicit dependence on $\eps$, but we do not fully understand their proof, which goes via a claimed equivalence with quantum SVMs.
Bellante and Zanero~\cite{BZ:qmpursuit} recently and independently used similar techniques as we use here for our Lasso upper bound (KP-trees and amplitude estimation) to give a polynomial quantum speedup for the classical matching-pursuit algorithm, which is a heuristic algorithm for the NP-hard problem of linear regression with a sparsity constraint, i.e., with an $\ell_0$-regularizer.

Another quantum approach for solving (unregularized) least-squares linear regression is based on the linear-systems algorithm of Harrow, Hassidim, and Lloyd~\cite{HHL09}. In this type of approach, the quantum algorithm very efficiently generates a solution vector $\theta$ \emph{as a quantum state} $\frac{1}{\norm{\theta}_2}\sum_i\theta_i\ket{i}$
(which is incomparable to our goal of returning $\theta$ as a classical vector). Chakraborty, Gily\'en, and Jeffery~\cite{CGJ18} used the framework of block-encodings to achieve this. Subsequently Gily\'en, Lloyd, and Tang~\cite{GLT18} obtained a ``dequantized'' classical algorithm for (unregularized) least-squares linear regression assuming length square sampling access to the input data, which again is incomparable to our setup.
The quantum algorithm was very recently improved with an $\ell_2$-regularizer by Chakraborty, Morolia, and Peduri~\cite{CMP:qlinereg}, thought still producing the final output as a quantum state rather than as a classical solution.

Norm-constrained linear regression is a special case of convex optimization. Quantum algorithms for various convex optimization problems have received much attention recently. For example, there has been a sequence of quantum algorithms for solving linear and semidefinite programs starting with Brand{\~a}o and Svore~\cite{brandao2016QSDPSpeedup,vAGGdW17,brandao2017QSDPSpeedupsLearning,vAG19,AG19zerosumgame}.
There have also been some polynomial speedups for matrix scaling~\cite{AGLNWW:scaling,gribling&nieuwboer:scaling} and for boosting in machine learning~\cite{arunachalam&maity:qboosting,izdebski&wolf:qboosting},
as well as some general speedups for converting membership oracles for a convex feasible set to separation oracles and optimization oracles~\cite{chakrabarti2018QuantumConvexOpt,apeldoorn2018ConvexOptUsingQuantumOracles,vApel20}.
On the other hand Garg, Kothari, Netrapalli, and Sherif~\cite{GKNS21noquantumspeedup} showed that the number of iterations for first-order algorithms for minimizing non-smooth convex functions cannot be significantly improved on a quantum computer; recently they generalized this result to higher-order algorithms~\cite{GKNS21porder}. Finally, there has also been work on quantum speedups for \emph{non-convex} problems, for instance on escaping from saddle points~\cite{ZLL:saddlepoints}.

\section{Preliminaries}

%\subsection{Some basic notation and theorems}

Throughout the paper, $d$ will always be the dimension of the ambient space $\mathbb{R}^d$, and $\log$ without a base will be the binary logarithm. It will be convenient for us to index entries of vectors starting from 0, so the entries~$x_i$ of a $d$-dimensional vector $x$ are indexed by $i\in\{0,\ldots,d-1\}=\intg_d$. $\cU_N=\mathcal{U}\{0,\ldots,N-1\}$ is the discrete uniform distribution over integers $0,1,2,\ldots,N-1$. %We will need a multivariate version of Taylor's theorem, whose proof we include for completeness.

\subsection{Computational model and quantum algorithms}
Our computational model is a classical computer (a classical random-access machine) that can invoke a quantum computer as a subroutine. The input is stored in quantum-readable read-only memory (a QROM), whose bits can be queried. The classical computer can also write bits to a quantum-readable classical-writable classical memory (a QRAM). The classical computer can send a description of a quantum circuit to the quantum computer; the quantum computer runs the circuit (which may include queries to the input bits stored in QROM and to the bits stored by the computer itself in the QRAM), measures the full final state in the computational basis, and returns the measurement outcome to the classical computer. In this model, an algorithm has time complexity $T$ if it uses at most $T$ elementary classical operations and quantum gates, quantum queries to the input bits stored in QROM, and quantum queries to the QRAM. The query complexity of an algorithm only measures the number of queries to the input stored in QROM. We call a (quantum) algorithm \emph{bounded-error} if (for every possible input) it returns a correct output with probability at least $9/10$. 

We will represent real numbers in computer memory using a number of bits of precision that is polylogarithmic in $d$, $N$, and $1/\eps$ (i.e., $\tilde{\mathcal{O}}(1)$ bits). This ensures all numbers are represented throughout our algorithms with negligible approximation error and we will ignore those errors later on for ease of presentation.

Below we state some important quantum algorithms that we will use as subroutines, starting with (an exact version of) Grover search and amplitude estimation.

\begin{theorem}[\cite{Grover96,BrassardHT98}]\label{thm:Grover}
    Let $f:\intg_d \rightarrow\{0, 1\}$ be a function that marks a set of elements $F = \{j \in \intg_d : f(j) = 1\}$ of known size $|F|$. Suppose that we have a quantum oracle $O_f$ such that $O_f:\ket{j}\ket{b} \rightarrow \ket{j}\ket{b\oplus f(j)}$. Then there exists a quantum algorithm that finds an index $j \in F$ with probability~1, using $\frac{\pi}{4}\sqrt{d/|F|}$ queries to $O_f$.
\end{theorem}

Note that we can use the above ``exact Grover'' repeatedly to find all elements of $F$ with probability~1, removing in each search the elements of $F$ already found in earlier searches. This even works if we only know an upper bound on $|F|$.

\begin{corollary}[\cite{bcwz:qerror}]\label{thm:findallsolutions}
    Let $f:\intg_d \rightarrow\{0, 1\}$ be a function that marks a set of elements $F = \{j \in \intg_d : f(j) = 1\}$. Suppose we know an upper bound $u$ on the size of $F$ and we have a quantum oracle $O_f$ such that $O_f:\ket{j}\ket{b} \rightarrow \ket{j}\ket{b\oplus f(j)}$. Then there exists a quantum algorithm that finds $F$ with probability~1, using $\frac{\pi}{2}\sqrt{du} + u$ queries to $O_f$.
\end{corollary}

\begin{proof}
Use the following algorithm: 
\begin{enumerate}
\item Set $S=\emptyset$
\item For $k=u$ downto 1 do:\\
\hspace*{1em}use Theorem~\ref{thm:Grover} on a modification $g$ of $f$, where $g(j)=0$ for all $j\in S$, assuming $|F|=k$;\\
\hspace*{1em}check that the returned value $j$ satisfies $f(j)=1$ by one more query; if so, add $j$ to $S$.
\end{enumerate}
Since we don't know $|F|$ exactly at the start, we are not guaranteed that each run of Grover finds another solution.
However, $k$ will always be an upper bound on the number of not-yet-found elements of $F$: either we found a new solution $j$ and we can reduce $k$ by 1 for that reason, or we did not find a new solution and then we know (by the correctness of the algorithm of Theorem~\ref{thm:Grover}) that the actual number of not-yet-found solutions was $<k$ and we are justified in reducing $k$ by~1. Hence at the end of the algorithm all elements of $F$ were found ($S=F$) with probability~1. The total number of queries is
$\displaystyle
\sum_{k=1}^u(\frac{\pi}{4}\sqrt{d/k}+1)\leq \frac{\pi}{4}\sqrt{d}\int_0^u \frac{1}{\sqrt{x}}dx \, +u=\frac{\pi}{2}\sqrt{du}+u$.
\end{proof}

\begin{theorem}[\cite{BHMT00}, Theorem 12]\label{thm:amplitude_estimation}
Given a natural number $M$ and access to an $(n + 1)$-qubit unitary $U$ satisfying
\[
U\ket{0^n}\ket{0}= \sqrt{a}\ket{\phi_1}\ket{1} +\sqrt{1-a}\ket{\phi_0}\ket{0},
\]
where $\ket{\phi_1}$ and $\ket{\phi_0}$ are arbitrary $n$-qubit states and $0 < a < 1$,
there exists a quantum algorithm that uses $\mathcal{O}(M)$ applications of $U$ and $U^\dagger$ and $\tilde{\mathcal{O}}(M)$ elementary gates, and outputs a state $\ket{\Lambda}$ such that after measuring that state, with probability $\geq 9/10$, the first register $\lambda$ of the outcome satisfies
\[
|a-\lambda| \leq \frac{\sqrt{(1-a)a}}{M}+\frac{1}{M^2}.
\]
\end{theorem}

%\begin{theorem}[\cite{DH96}, Theorem~1]\label{thm:Qmin}
%    Let $f:\intg_d\rightarrow \mathbb{R}$. Suppose that we have a quantum oracle $O_f$ such that $O_f\ket{j}\ket{0} = \ket{j}\ket{f(j)}$. There exists a quantum algorithm that uses $\mathcal{O}(\sqrt{d})$ queries to $O_f$,  $\tilde{\mathcal{O}}(\sqrt{d})$ elementary gates, and finds an index $y$ such that $f(j)$ is the minimum with probability $\geq0.999$.
%\end{theorem}

The following is a modified version of quantum minimum-finding, which in its basic form is due to H{\o}yer and D\"urr~\cite{DH96}. Our proof uses a result from~\cite{vAGGdW17}, see Appendix~\ref{app:proof_approx_min}.

\begin{theorem}[min-finding with an approximate unitary]\label{thm:min_finding_approx} Let $\delta_1, \delta_2, \epsilon \in (0,1)$, $v_0,\ldots,v_{d-1}\in\mathbb{R}$. Suppose we have a unitary $\tilde{A}$ that maps $\ket{j}\ket{0}\rightarrow\ket{j}\ket{\Lambda_j}$ such that for every $j\in\intg_d$, after measuring the state $\ket{\Lambda_j}$, with probability $\geq 1-\delta_2$ the first register $\lambda$ of the measurement outcome satisfies $|\lambda-v_j|\leq \epsilon$. There exists a quantum algorithm that finds an index $j$ such that $v_j\leq \min_{k\in\intg_d}v_k+2\epsilon$ with probability $\geq 1-\delta_1-1000\log(1/\delta_1)\cdot\sqrt{2d\delta_2}$, using $1000\sqrt{d}\cdot \log(1/\delta_1)$ applications of $\tilde{A}$ and $\tilde{A}^\dagger$, and $\mathcal{\tilde{O}}(\sqrt{d})$ elementary gates. In particular, if $\delta_2 \leq \delta_1^2/(2000000d\log(1/\delta_1))$, then the above algorithm finds such a $j$ with probability $\geq 1-2\delta_1$.
\end{theorem}

\subsection{Expected and empirical loss}\label{sec:loss}
Let sample set $S = \{(x_i, y_i)\}_{i=0}^{N-1}$ be a set of i.i.d.\ samples from $\mathbb{R}^d \times \mathbb{R}$, drawn according to an unknown distribution~$\mathcal{D}$. A \emph{hypothesis} is a function $h : \mathbb{R}^d \rightarrow \mathbb{R}$, and $\mathcal{H}$ denotes a set of hypotheses.
To measure the performance of the prediction, we use a convex loss function $\ell:\mathbb{R}^2\rightarrow \mathbb{R}$. 
The \emph{expected loss} of $h$ with respect to $\mathcal{D}$ is denoted by $L_\mathcal{D}(h)=\E_{(x,y)\sim \mathcal{D}}[\ell(h(x),y)]$, and the \emph{empirical loss} of $h$ with respect to $S$ is denoted by $L_S(h)=\frac{1}{N}\sum\limits_{i\in\intg_N}\ell(h(x_i),y_i)$.

\begin{definition}
Let $\epsilon >0$. An $h\in \cH$ is an $\epsilon$-minimizer over $\cH$ with respect to distribution $\mathcal{D}$ if
\[
L_\mathcal{D}(h)-\min_{h'\in \cH} L_\mathcal{D}(h') \leq \epsilon.
\]
\end{definition}

\begin{definition}
Let $\epsilon >0$. An $h\in \cH$ is an $\epsilon$-minimizer over $\cH$ with respect to sample set $S$ if
\[
L_S(h)-\min_{h'\in \cH} L_S(h') \leq \epsilon.
\]
\end{definition}

\subsection{Linear regression problems and their classical and quantum setup}\label{sec:reg_model}

In linear regression problems, the hypothesis class is the set of linear functions on $\mathbb{R}^d$. The goal is to find a vector $\theta$ for which the corresponding hypothesis $\langle \theta, x\rangle$ provides a good prediction of the target $y$. One of the most natural choices for regression problems is the squared loss
$$
\ell(\hat{y},y)=(\hat{y}-y)^2.
$$ 
We can instantiate the expected and empirical losses as a function of $\theta$ using the squared loss:
$$
L_{\mathcal{D}}(\theta) = \E_{(x, y)\sim \mathcal{D}}[\ell(\inProd{\theta}{x}, y)]=\E_{(x, y)\sim \mathcal{D}}[(\inProd{\theta}{x}-y)^2],
$$ 
$$
L_S(\theta)=\frac{1}{N}\sum\limits_{i\in\intg_N} \ell(\inProd{\theta}{x},y_i)=\frac{1}{N}\sum\limits_{i\in\intg_N} (\inProd{\theta}{x}-y_i)^2.
$$ 
We also write the empirical loss as $L_S(\theta)=\frac{1}{N}\|X\theta-y\|^2_2$, where matrix entry $X_{ij}$ is the $j$th entry of the vector $x_i$, and $y$ is the $N$-dimensional vector with entries $y_i$.
As we will see below, if the instances in the sample set are chosen i.i.d.\ according to $\mathcal{D}$, and $N$ is sufficiently large, then 
$L_S(\theta)$ and $L_{\mathcal{D}}(\theta)$ are typically close by the law of large numbers. 

In the quantum case, we assume the sample set $S$ is stored in a QROM, which we can access by means of queries to the oracles $O_X:\ket{i}\ket{j}\ket{0}\rightarrow \ket{i}\ket{j}\ket{X_{ij}}$ and $O_y:\ket{i}\ket{0}\rightarrow \ket{i}\ket{y_{i}}$. 

\subsubsection{Lasso}

The \emph{least absolute shrinkage and selection operator}, or \emph{Lasso}, is a special case of linear regression with a norm constraint on the vector $\theta$: it restricts solutions to the unit $\ell_1$-ball, which we denote by $B_1^d$. For the purpose of normalization, we require that every sample $(x,y)$ satisfies $\|x\|_\infty \leq 1$ and $|y|\leq 1$.% 
\footnote{Note that if $\theta\in B_1^d$ and $\|x\|_\infty \leq 1$, then $|\inProd{\theta}{x}|\leq 1$  by H{\"o}lder's inequality.}
The goal is to find a $\theta\in B^d_1$ that (approximately) minimizes the expected loss. Since the expected loss is not directly accessible, we instead find an approximate minimizer of the empirical loss. Mohri, Rostamizadeh, and Talwalkar~\cite{MRT18} showed that with high probability, an approximate minimizer for \emph{empirical} loss is also a good approximate minimizer for \emph{expected} loss.

\begin{theorem}[\cite{MRT18}, Theorem~11.16]\label{thm:Rademacher_Lasso}
Let $\mathcal{D}$ be an unknown distribution over $[-1,1]^d\times [-1,1]$ and $S=\{(x_i,y_i)\}_{i=0}^{N-1}$ be a sample set containing $N$ i.i.d.\ samples from $\mathcal{D}$. Then, for each $\delta > 0$, with probability $\geq 1 -\delta$ over the choice of $S$, the following holds for all $\theta\in B^d_1$: 
 \begin{align*}
L_\mathcal{D}(\theta)-L_S(\theta)\leq 4\sqrt{\frac{2\log (2d)}{N}}+4\sqrt{\frac{\log (1/\delta)}{2N}}.
 \end{align*}
\end{theorem}

This theorem implies that if  $N=c\log(d/\delta)/\eps^2$ for sufficiently large constant $c$, then finding (with error probability $\leq \delta$) an $\eps$-minimizer for the \emph{empirical} loss $L_S$, implies finding (with error probability $\leq 2\delta$ taken both over the randomness of the algorithm and the choice of the sample~$S$) a $2\eps$-minimizer for the \emph{expected} loss~$L_\mathcal{D}$.

\subsubsection{Ridge}

Another special case of linear regression with a norm constraint is \emph{Ridge}, which restricts solutions to the unit $\ell_2$-ball $B_2^d$. For the purpose of normalization, we now require that every sample $(x,y)$ satisfies $\|x\|_2\leq 1$ and $|y|\leq 1$. Similarly to the Lasso case, Mohri, Rostamizadeh, and Talwalkar~\cite{MRT18} showed that with high probability, an approximate minimizer for the empirical loss is also a good approximate minimizer for the expected loss.

\begin{theorem}[\cite{MRT18}, Theorem~11.11]\label{thm:Rademacher_Ridge}
Let $\mathcal{D}$ be an unknown distribution over $B_2^d\times [-1,1]$ and $S=\{(x_i,y_i)\}_{i=0}^{N-1}$ be a sample set containing $N$ i.i.d.\ samples from $\mathcal{D}$. Then, for each $\delta  > 0$, with probability $\geq 1 -\delta$ over the choice of $S$, the following holds for all $\theta\in B^d_2$: 
 \begin{align*}
L_\mathcal{D}(\theta)-L_S(\theta)\leq 8\sqrt{\frac{1}{N}}+4\sqrt{\frac{\log (1/\delta)}{2N}}.
 \end{align*}
\end{theorem}

\subsection{The KP-tree data structure and efficient state preparation}
Kerenidis and Prakash~\cite{Pra14,KP16} gave a quantum-accessible classical data structure to store a vector $\theta$ with support $t$ (i.e., $t$ nonzero entries) to enable efficient preparation of the state 
\[
\ket{\theta}=\sum\limits_{j\in \mathbb{Z}_d}\sqrt{\frac{|\theta_j|}{\|\theta\|_1}}\ket{j}\ket{sign(\theta_j)}. 
\]
In this subsection, we modify their data structure such that for arbitrary $a,b \in \mathbb{R}$ and $j\in \mathbb{Z}_d$, we can  efficiently update a data structure for the vector $\theta$ to a data structure for the vector $a\theta+be_j$, without having to individually update all nonzero entries of the vector. We call this data structure a ``KP-tree'' (or $KP_\theta$ if we're storing vector $\theta$) in order to credit Kerenidis and Prakash. 

\begin{definition}[KP-tree]\label{def:KPtree}
Let $\theta\in\mathbb{R}^d$ have support $t$. We define a KP-tree $KP_\theta$ of $\theta$ as follows:
\begin{itemize}
    \item $KP_\theta$ is a rooted binary tree with depth $\lceil\log d\rceil$ and with $\mathcal{O}(t\log d)$ vertices.
    \item The root stores a scalar $A\in \mathbb{R}\setminus\{0\}$ and the support $t$ of $\theta$.
    \item Each edge of the tree is labelled by a bit.
    \item For each $j\in supp(\theta)$, there is one corresponding leaf storing $\frac{\theta_j}{A}$. The number of leaves is $t$.
    \item The bits on the edges of the path from the root to the leaf corresponding to the $j^{th}$ entry of $\theta$, form the binary description of $j$.
    \item Each intermediate node stores the sum of its children's absolute values.
\end{itemize}
For $\ell\in \intg_{\lceil\log d\rceil}$ and $j\in \intg_{2^\ell}$, we define $KP_\theta(\ell,j)$ as the value of the $j^{th}$ node in the $\ell^{th}$ layer, i.e., the value stored in the node that we can reach by the path according to the binary representation of $j$ from the root. Also, we let $KP_\theta(0,0)$ be the sum of all absolute values stored in the leaves.
If there is no corresponding $j^{th}$ node in the $\ell^{th}$ layer (that is, we cannot reach a node by the path according to the binary representation of $j$ from the root), then $KP_\theta(\ell,j)$ is defined as $0$. Note that both the numbering of the layer and the numbering of nodes start from $0$. In the special case where $\theta$ is the all-0 vector, the corresponding tree will just have a root node with $t=0$.
\end{definition}

\begin{figure}[!htb]
   \begin{minipage}{0.49\textwidth}
     \centering
\tikzset{every tree node/.style={minimum width=0.2em,draw,circle},
         blank/.style={draw=none},
         edge from parent/.style=
         {draw, edge from parent path={(\tikzparentnode) -- (\tikzchildnode)}},
         level distance=1cm}
\begin{tikzpicture}[scale=0.90]
\Tree
[.20,3     
    [.1 
    \edge[blank]; \node[blank]{};
    \edge[]; [.1
             \edge[]; {1}
             \edge[blank]; \node[blank]{};
         ]
    ]
    [.5  
    \edge[blank]; \node[blank]{};
    \edge[]; [.5
             \edge[]; {2}  \edge[]; {-3}
         ]
    ]
]
\end{tikzpicture}
\qquad
\tikzset{every tree node/.style={minimum width=0.2em,draw,circle},
         blank/.style={draw=none},
         edge from parent/.style=
         {draw, edge from parent path={(\tikzparentnode) -- (\tikzchildnode)}},
         level distance=1cm}
\begin{tikzpicture}[scale=0.90]
\Tree
[.10,3     
    [.2 
    \edge[blank]; \node[blank]{};
    \edge[]; [.2
             \edge[]; {2}
             \edge[blank]; \node[blank]{};
         ]
    ]
    [.10  
    \edge[blank]; \node[blank]{};
    \edge[]; [.10
             \edge[]; {4}  \edge[]; {-6}
         ]
    ]
]
\end{tikzpicture}
\caption{Each of the above two binary trees represents the vector $\theta=20e_2+40e_6-60e_7$. If we see the second layer of KP$_\theta$ on the right-hand side, $KP_\theta(2,0)=0$, $KP_\theta(2,1)=2$, $KP_\theta(2,2)=0$, and $KP_\theta(2,3)=10$.} \label{fig:KPTree}
   \end{minipage}\hfill
   \begin{minipage}{0.49\textwidth}
     \centering
\tikzset{every tree node/.style={minimum width=0.2em,draw,circle},
         blank/.style={draw=none},
         edge from parent/.style=
         {draw, edge from parent path={(\tikzparentnode) -- (\tikzchildnode)}},
         level distance=1cm}
\begin{tikzpicture}[scale=0.90]
\Tree
[.10,3     
    [.2 
    \edge[blank]; \node[blank]{};
    \edge[]; [.2
             \edge[]; {2}
             \edge[blank]; \node[blank]{};
         ]
    ]
    [.5  
    \edge[blank]; \node[blank]{};
    \edge[]; [.5
             \edge[]; {2}  \edge[]; {-3}
         ]
    ]
]
\end{tikzpicture}
\qquad
\tikzset{every tree node/.style={minimum width=0.2em,draw,circle},
         blank/.style={draw=none},
         edge from parent/.style=
         {draw, edge from parent path={(\tikzparentnode) -- (\tikzchildnode)}},
         level distance=1cm}
\begin{tikzpicture}[scale=0.90]
\Tree
[.8,4     
    [.4 
    \edge[blank]; \node[blank]{};
    \edge[]; [.4
             \edge[]; {2}
             \edge[]; {2}
         ]
    ]
    [.5  
    \edge[blank]; \node[blank]{};
    \edge[]; [.5
             \edge[]; {2}  \edge[]; {-3}
         ]
    ]
]
\end{tikzpicture}
\caption{The update rule: we update the vector $\theta=20e_2+20e_6-30e_7$ to the new vector $\frac{4}{5}\theta+16e_4$ by updating the scalar in the root to $\frac{4}{5}\cdot 10=8$, adding a new leaf with value $16/8=2$, recomputing the values of the intermediate nodes between the root and the leaf, and updating the support number to $4$.} \label{fig:KPTreeUpdate}
   \end{minipage}
\end{figure}

\begin{theorem}For each $j\in \mathbb{Z}_d$, one can read the number $\theta_j$ by reading at most $\poly\log d$ nodes of $KP_\theta$ and by using $\poly\log d$ many (classical) elementary operations.
\end{theorem}

\begin{proof}
Read the scalar $A$ stored in the root. Choose the path according to the binary representation of $j$. If the chosen path reaches a leaf, then read the value $v$ at that leaf and output $vA$. If it does not reach a leaf, output $0$. The total cost is at most $\poly\log d$ because the depth of $KP_\theta$ is $\lceil\log d\rceil$.

From the fourth bullet of Definition~\ref{def:KPtree}, if $j\in supp(\theta)$, then the corresponding leaf $j$ stores $\theta_j/A$ and hence the output is $(\theta_j/A)\cdot A=\theta_j$. On the other hand, if $j\notin supp(\theta)$, then we do not reach a leaf and know $\theta_j=0$.
\end{proof}

\begin{theorem}\label{thm:KP_update_cost}
Given a KP-tree $KP_\theta$, $j\in\mathbb{Z}_d$,  and numbers $a\in \mathbb{R}\setminus \{0\}$ and $b\in \mathbb{R}$, we can update $KP_{\theta}$ to  $KP_{a\theta+be_j}$ by using $\poly\log d$ elementary operations and by modifying $\poly\log d$ many values stored in the nodes of $KP_\theta$.
\end{theorem}

\begin{proof}
Read the scalar $A$ and support $t$ stored in the root. 

If there does not exist a leaf for the entry $j$, then add a new leaf for the entry $j$ and a path according to its binary representation. Now update the stored value in the leaf $j$ to $b/(aA)$. After that, update the stored values for all nodes on the path from the root to the leaf for the entry $j$. Update the scalar in the root to $aA$, and if $b \neq 0$, update the support value in the root to $t+1$.

If, instead, there already existed a leaf for the entry $j$, then read the value $v$ stored in the leaf~$j$, update the value stored in the leaf~$j$ to $v'=v+b/(aA)$, and then update the stored values for all nodes on the path from the root to the leaf $j$ for the entry $j$, and update the scalar to $aA$. After that, check the value $v'$ stored in the leaf for the entry $j$; if $v'=0$, then remove all nodes storing the value $0$ from the leaf $j$ to the root, and update the support value at the root to $t-1$.
\end{proof}

\begin{theorem}\label{thm:StatePrepare}
Suppose we have a KP-tree $KP_\theta$ of vector $\theta$, and suppose  we can make quantum queries to a unitary $O_{KP_\theta}$ that maps $\ket{\ell,k}\ket{0}\rightarrow \ket{\ell,k}\ket{KP_\theta(\ell,k)}$. Then one can prepare the state $\ket{\theta}=\sum\limits_{j\in \mathbb{Z}_d}\sqrt{\frac{|\theta_j|}{\|\theta\|_1}}\ket{j}\ket{sign(\theta_j)}$ up to negligible error\footnote{By this we mean an error smaller than an arbitrary polynomial in the input length.} by using $\poly\log d$ queries to $O_{KP_\theta}$ and $O^\dagger_{KP_\theta}$, and $\mathcal{\tilde{O}}(1)$ elementary gates.
\end{theorem}

\begin{proof}
For simplicity and without loss of the generality, we assume $\log d$ is a natural number. Define the two-controlled rotation unitary as for each $a,b\in\mathbb{R}$
\[
U_{2CR}: \ket{a}\ket{b} \ket{0} \rightarrow 
    \begin{dcases}
    \ket{a}\ket{b}(\frac{1}{\sqrt{2}}\ket{0}+\frac{1}{\sqrt{2}}\ket{1}), & \text{if } a=b=0,\\
    \ket{a}\ket{b}(\sqrt{\frac{|a|}{|a|+|b|}}\ket{0}+\sqrt{\frac{|b|}{|a|+|b|}}\ket{1}), &\text{otherwise,}
\end{dcases}
\]
which can be implemented up to negligibly small error by $\mathcal{\tilde{O}}(1)$ elementary gates. Also, define the children-reading gate as $U_C:\ket{\ell}\ket{k}\ket{0}^{\otimes 2} \rightarrow \ket{\ell}\ket{k}\ket{l_{\ell,k}}\ket{r_{\ell,k}}$, where the left child $l_{\ell,k}=KP_\theta(\ell+1,2k)$ and the right child $r_{\ell,k}=KP_\theta(\ell+1,2k+1)$; this can be implemented by using two queries to $O_{KP_\theta}$ and $\mathcal{\tilde{O}}(1)$ elementary gates. Last, define the sign gate $U_s:\ket{j}\ket{0}\rightarrow \ket{j}\ket{sign(\theta_{j})}$, which can be implemented by using two queries to $O_{KP_\theta}$, $O^\dagger_{KP_\theta}$, and $\mathcal{{O}}(1)$ elementary gates. 

To prepare $\ket{\theta}$, we first prepare the state $\ket{KP_\theta^0}=\ket{0}$, and
%\[
%\ket{KP_\theta^0}=\frac{1}{\sqrt{|KP_\theta(1,0)|+|KP_\theta(1,1)|}}(\sqrt{KP_\theta(1,0)}\ket{0}+\sqrt{KP_\theta(1,1)}\ket{1})
%\]
%by preparing $\ket{0}\ket{0}\ket{0}\ket{0}\ket{0}$, applying $U_C$ on the first four registers, applying $U_{2CR}$ on the last three registers, uncomputing the third and the fourth registers, and discarding the first four registers.\rnote{If the base step of an inductive argument uses the same argument as the inductive step, then you can usually make the argument simpler by using a simpler earlier step as the base case (ie just prepare the root node), to avoid the repetition. Your $KP^0$ is for the 1st layer, let $KP^0$ be for the 0th layer instead} 
for the purpose of induction, suppose we can prepare the state 
$$
\ket{KP_\theta^{\ell}}=\frac{1}{\sqrt{KP_\theta(0,0)}}\sum\limits_{k=0}^{2^{\ell}-1} \sqrt{|KP_\theta(\ell,k)|}\ket{k},
$$
%where $B=|l_{0,0}|+|r_{0,0}|=|KP_\theta(1,0)|+|KP_\theta(1,1)|$\rnote{Is this state normalized? I guess it is, maybe introduce and explain this $B$ already for $KP^0$?}.
where $KP_\theta(0,0)$ is the sum of all absolute values stored in the leaves and hence $KP_\theta(0,0)=\sum\limits_{k=0}^{2^{\ell}-1} |KP_\theta(\ell,k)|$. We prepare the state $\ket{\ell}\ket{KP_\theta^{\ell}}\ket{0}^{\otimes 2}\ket{0}$, apply $U_C$ on the first four registers, and apply $U_{2CR}$ on the last three registers to get
\begin{align*}
&\ket{\ell}\frac{1}{\sqrt{KP_\theta(0,0)}}\sum\limits_{k=0}^{2^{\ell}-1} \sqrt{|KP_\theta(\ell,k)|}\ket{k}\ket{l_{\ell,k}}\ket{r_{\ell,k}}\big(\frac{\sqrt{|l_{\ell,k}|}}{\sqrt{|l_{\ell,k}|+|r_{\ell,k}|}}\ket{0}+\frac{\sqrt{|r_{\ell,k}|}}{\sqrt{|r_{\ell,k}|+|r_{\ell,k}|}}\ket{1}\big)\\
=& \ket{\ell}\frac{1}{\sqrt{KP_\theta(0,0)}}\sum\limits_{k=0}^{2^{\ell}-1}\ket{k}\ket{l_{\ell,k}}\ket{r_{\ell,k}}\big(\sqrt{|l_{\ell,k}|}\ket{0}+\sqrt{|r_{\ell,k}|}\ket{1}\big),
\end{align*}
where the equation holds because $|KP_\theta(\ell,k)|=|KP_\theta(\ell+1,2k)|+|KP_\theta(\ell+1,2k+1)|=|l_{\ell,k}|+|r_{\ell,k}|$, from the sixth bullet of Definition~\ref{def:KPtree}.
Uncomputing the third and fourth registers, and discarding the first, third, and fourth registers, we get
\begin{align*}
&\frac{1}{\sqrt{KP_\theta(0,0)}}\sum\limits_{k=0}^{2^{\ell}-1}\ket{k}\big(\sqrt{|l_{\ell,k}|}\ket{0}+\sqrt{|r_{\ell,k}|}\ket{1}\big)\\
=& \frac{1}{\sqrt{KP_\theta(0,0)}}\sum\limits_{k=0}^{2^{\ell}-1}\big(\sqrt{|l_{\ell,k}|}\ket{k}\ket{0}+\sqrt{|r_{\ell,k}|}\ket{k}\ket{1}\big)\\
=& \frac{1}{\sqrt{KP_\theta(0,0)}}\sum\limits_{k=0}^{2^{\ell+1}-1}\sqrt{|KP_\theta(\ell+1,k)|}\ket{k}=\ket{KP_\theta^{\ell+1}}.
\end{align*}
Therefore, iterating the above process for $\log d$ times, we can prepare the state \[
\ket{KP_\theta^{\log d}}=\frac{1}{\sqrt{KP_\theta(0,0)}}\sum\limits_{k=0}^{d-1} \sqrt{|KP_\theta(\log d,k)|}\ket{k}=\sum\limits_{j\in \mathbb{Z}_d}\sqrt{\frac{|\theta_j|}{\|\theta\|_1}}\ket{j},
\]
where the last equation follows from the fourth and sixth bullets of Definition~\ref{def:KPtree}. To obtain $\ket{\theta}$, we prepare $\sum\limits_{j\in \mathbb{Z}_d}\sqrt{\frac{|\theta_j|}{\|\theta\|_1}}\ket{j}\ket{0}$ and apply $U_s$.

 There are $\log d$ layers, and each layer only uses $\mathcal{O}(1)$ queries to $O_{KP_\theta}$, $O^\dagger_{KP_\theta}$ and $\mathcal{\tilde{O}}(1)$ elementary gates. Hence $\mathcal{O}(\log d)$ queries to $O_{KP_\theta}$, $O^\dagger_{KP_\theta}$, and $\mathcal{\tilde{O}}(1)$ other gates suffice to prepare $\ket{\theta}$.
\end{proof}

To implement $O_{KP_\theta}$ and $O^\dagger_{KP_\theta}$ in the above theorem, we use QRAM to store $KP_\theta$, then we can make quantum queries to the bits of the data structure directly. Or, if we want to avoid QRAM altogether, then we can use the following theorem with $\tilde{\mathcal{O}}(s)$ extra cost in circuit size for each query, where $s$ is the sparsity of the bitstring that represents $KP_\theta$. From Definition~\ref{def:KPtree} we can see that the number of bits is $s=\mathcal{\tilde{O}}(t\log d)$, where $t$ is the sparsity of $\theta$.

\begin{theorem}\label{thm:QRAM_free}
Suppose $p,s\in\mathbb{N}$ and $D\in\{0,1\}^p$ is a bit string with sparsity $s$ (i.e., the number of 1s in $D$ is $\leq s$), then for each $b\in\{0,1\}$ and $k\in \intg_p$, we can implement the unitary $U_D:\ket{k,b}\rightarrow \ket{k,b\oplus D_k}$ using ${\mathcal{O}}(s\log p)$ elementary gates.
\end{theorem}

\begin{proof}
For every $i\in \intg_p$, we define the controlled bit-reading unitary $U_i$ as for each $b\in \{0,1\}$ and $k\in\{0,1\}^p$
\begin{align*}
    U_i: \ket{k}\ket{b} \rightarrow 
    \begin{dcases}
    \ket{k}\ket{b\oplus 1}, & \text{if } k=i,\\
    \ket{k}\ket{b}, &\text{otherwise,}
\end{dcases}
\end{align*}
which can be implemented using $\mathcal{{O}}(\log p)$ elementary gates. Observing that $U_D= \Pi_{i:D_i=1} U_i$, we can therefore implement $U_D$ using $\mathcal{{O}}(s\cdot \log p)$ elementary gates.
\end{proof}

\section{Quantum Algorithm for Lasso}

\subsection{The classical Frank-Wolfe algorithm}
Below is a description of the Frank-Wolfe algorithm with approximate linear solvers.
For now this is for an arbitrary convex objective function~$L$ and arbitrary compact convex domain $\mathcal{X}$ of feasible solutions; for Lasso we will later instantiate these to the quadratic loss function and $\ell_1$-ball, respectively.
Frank-Wolfe finds an $\epsilon$-approximate solution to a convex optimization problem, using $O(1/\epsilon)$ iterations. It is a first-order method: each iteration assumes access to the gradient of the objective function at the current point. The algorithm considers the linearization of the objective function, and moves towards a minimizer of this linear function without ever leaving the domain $\mathcal{X}$ (in contrast to for instance projected gradient descent). 

\begin{algorithm}[hbt]
\SetKwData{Left}{left}\SetKwData{This}{this}\SetKwData{Up}{up}
\SetKwFunction{Union}{Union}\SetKwFunction{FindCompress}{FindCompress}
\SetKwInOut{Input}{input}\SetKwInOut{Output}{output}

\Input{number of iterations $T >0$; convex differentiable function $L$; compact convex domain $\mathcal{X}$;}
Let $C_{L}$ be the curvature constant of $L$;\\
Let $\theta^0$ be an arbitrary point in $\mathcal{X}$; 
\SetAlgoLined

  \For{$t\leftarrow 0$ \KwTo $T$}{
   $\tau_t= \frac{2}{t+2}$\;
   find $s\in \mathcal{X}$ such that $\langle s,\nabla L(\theta^{t})\rangle\leq \min\limits_{s'\in \mathcal{X}}\langle s',\nabla L(\theta^{t})\rangle + \frac{ \tau_t C_{L}}{4}$\;
  $\theta^{t+1}=(1-\tau_t) \theta^{t}+\tau_t s$\;
  }
  \Output{$\theta^T$;}
 \caption{The Frank-Wolfe algorithm with approximate linear subproblems}
 \label{Alg:FW}
\end{algorithm}

The convergence rate of the Frank-Wolfe algorithm is affected by the ``non-linearity'' of the objective function $L$, as measured by the curvature constant $C_{L}$:

\begin{definition}
 The curvature constant $C_{L}$ of a convex and differentiable function $L:\mathbb{R}^d\rightarrow \mathbb{R}$ with respect
to a convex domain $\mathcal{X}$ is defined as
\[
C_{L}\equiv \sup\limits_{\substack{x,s\in \mathcal{X}, \gamma \in [0,1],\\ y=x+\gamma(s-x)}} \frac{2}{\gamma^2}(L(y)-L(x)-\langle \nabla L(x),(y-x)\rangle).
\]
\end{definition}

Next we give an upper bound for the curvature constant of the empirical loss function for Lasso.

\begin{theorem}\label{thm:curve_bound} Let $S=\{(x_i,y_i)\}_{i=0}^{N-1}$ with all entries of $x_i$ and $y_i$ in $[-1,1]$. Then the curvature constant $C_{L_S}$ of $L_S$ with respect to $B^d_1$ is $\leq8$.
\end{theorem}

\begin{proof}
We know
\[
L_S(\theta)=\frac{1}{N}\|X\theta-y\|_2^2=\frac{(X\theta-y)^T(X\theta-y)}{N}=\frac{\theta^TX^TX\theta-y^TX\theta-\theta^TX^Ty+y^Ty}{N},
\]
which implies the Hessian of $L_S$ is $\nabla^2 L_S(z)=\frac{2X^TX}{N}$, independent of~$z$.
By replacing sup by max because the domain is compact, we have
\begin{eqnarray*}
C_{L_S} & = \max\limits_{\substack{x,s\in \mathcal{X}, \gamma \in [0,1],\\ y=x+\gamma(s-x)}} \frac{2}{\gamma^2}(L_S(y)-L_S(x)-\langle \nabla L_S(x),(y-x)\rangle)\\ 
& = \max\limits_{x,s\in \mathcal{X}, \gamma\in [0,1]} \langle (s-x),\nabla^2 L_S\cdot(s-x) \rangle= \max\limits_{x,s\in \mathcal{X}} \frac{2}{N} \|X(s-x)\|^2_2.
\end{eqnarray*}
Each coefficient of $X$ is at most $1$ in absolute value, and $s-x\in 2B^d_1$, hence each entry of the vector $X(s-x)$ has magnitude at most~$2$. Therefore $\max\limits_{x,y\in B^d_1} \frac{2}{N} \|X(s-x)\|^2_2$ is at most $8$.
\end{proof}

The original Frank-Wolfe algorithm~\cite{FW56} assumed that the minimization to determine the direction-of-change~$s$ was done exactly, without the additive error term~$\tau_t C_{L_S}/4$ that we wrote in Algorithm~\ref{Alg:FW}.
However, the following theorem, due to Jaggi~\cite{Jaggi13}, shows that solving approximate linear subproblems is sufficient for the Frank-Wolfe algorithm to converge at an $O(C_{L_S}/T)$ rate, which means one can find an $\eps$-approximate solution with $T=O(C_{L_S}/\eps)$ iterations.

\begin{theorem}[\cite{Jaggi13}, Theorem~1]\label{thm:FW}
 For each iteration $t\geq 1$, the corresponding $\theta^{t}$ of Algorithm~\ref{Alg:FW} satisfies
 \[
L_S(\theta^{t})- \min\limits_{\theta'\in B^d_1}L_S(\theta')\leq \frac{3C_{L_S} }{t+2}. 
 \]
\end{theorem}

\subsection{Approximating the quadratic loss function and entries of its gradient}

In this subsection, we give a quantum algorithm to estimate the quadratic loss function $L_S(\theta)$ and entries of its gradient, given query access to entries of the vectors in $S=\{(x_i,y_i)\}_{i=0}^{N-1}$ and given a KP-tree for $\theta \in B^d_1$.
One can estimate these numbers with additive error $\beta$ in time roughly $1/\beta$.

We start with estimating entries of the gradient of the loss function at a given $\theta$:

\begin{theorem}\label{thm:nabla_error_new}
Let $\theta\in B^d_1$, and $\beta, \delta >0$. Suppose we have a KP-tree $KP_\theta$ of vector $\theta$ and can make quantum queries to $O_{KP_\theta}: \ket{\ell,k}\ket{0}\rightarrow \ket{\ell,k}\ket{KP_\theta(\ell,k)}$. One can implement $\tilde{U}_{\nabla L_S}: \ket{j}\ket{0}\rightarrow \ket{j}\ket{\Lambda}$ such that for all $j\in \mathbb{Z}_d$, after measuring the state $\ket{\Lambda}$, with probability $\geq 1-\delta$ the first register $\lambda$ of the outcome will satisfy
$|\lambda-\nabla_j L_S(\theta)|\leq \beta,$
 by using $\mathcal{\tilde{O}}(\frac{\log(1/\delta)}{\beta})$  applications of $O_X$, $O^\dagger_X$, $O_y$, $O^\dagger_y$, $O_{KP_\theta}$, $O^\dagger_{KP_\theta}$, and elementary gates.
\end{theorem}

\begin{proof}
Fix $j$ in the following proof.
Note that
\begin{align}
\nabla_j L_S(\theta) =\frac{2}{N} (X^T(X\theta-y))_j=\frac{2}{N}\sum\limits_{i\in\intg_N} X_{ij}\cdot (X\theta-y)_i=\frac{2}{N}\sum\limits_{i\in\intg_N}\sum\limits_{k\in\intg_d} X_{ij}X_{ik}\theta_{k}-\frac{2}{N}\sum\limits_{i\in\intg_N} X_{ij}y_i.\label{eq:nabla_lasso}
\end{align}
We will show how to estimate both terms of the right-hand side of Equation~(\ref{eq:nabla_lasso}). 
Define the positive-controlled rotation such that for each $a\in\mathbb{R}$
\[
U_{CR^+}: \ket{a} \ket{0} \rightarrow \begin{dcases}
    \ket{a}(\sqrt{a}\ket{1}+\sqrt{1-a}\ket{0}),& \text{if } a\in(0,1]\\
    \ket{a}\ket{0},              & \text{otherwise.}
\end{dcases}
\]
This can be implemented up to negligibly small error by $\tilde{\mathcal{O}}(1)$ elementary gates. Also, by Theorem~\ref{thm:StatePrepare}, we can implement $U_{\theta}:\ket{0}\ket{0}\rightarrow \ket{\theta}$, where $\ket{\theta}=\frac{1}{\sqrt{\|\theta\|_1}}\sum\limits_{j\in\mathbb{Z}_d} \sqrt{|\theta_j|}\ket{j}\ket{sign(\theta_j)}$, using $\mathcal{\tilde{O}}(1)$ queries to $O_{KP_\theta}$, $O_{KP_\theta}^\dagger$, and elementary gates. We will also use $U_u:\ket{i}\ket{j}\ket{k}\ket{s}\ket{0}\rightarrow \ket{i}\ket{j}\ket{k}\ket{s}\ket{s\cdot X_{ij}\cdot X_{ik}}$, where the last register is the product of three numbers; this $U_u$ can be implemented (with negligible error) via $\mathcal{O}(1)$ queries to $O_X$, $O^\dagger_X$, and $\tilde{\mathcal{O}}(1)$ elementary gates. 

First we estimate the first term $\frac{2}{N}\sum\limits_{i\in \intg_N}\sum\limits_{k\in\intg_d} X_{ij}X_{ik}\theta_{k}$ of Equation~(\ref{eq:nabla_lasso}) with additive error $\beta/2$.
Generating the state 
$$
\frac{1}{\sqrt{N}}\sum\limits_{i\in\intg_N}\ket{i}\ket{j}\ket{0}\ket{0}\ket{0}\ket{0},
$$
and applying $U_\theta$ on the third and fourth registers, $U_u$ on the first five registers, and $U_{CR^+}$ on the last two registers, with $s_k=sign(\theta_k)$, we get
\begin{align*}
&\frac{1}{\sqrt{N\|\theta\|_1}}\Big(\sum\limits_{\substack{i\in\mathbb{Z}_N,k\in \mathbb{Z}_d\\s_kX_{ij}X_{ik}>0}}\sqrt{|\theta_k|}\ket{i,j,k,s_k,s_kX_{ij}X_{ik}}(\sqrt{s_kX_{ij}X_{ik}}\ket{1}+\sqrt{1-s_kX_{ij}X_{ik}}\ket{0})\\
&+\sum\limits_{\substack{i\in\mathbb{Z}_N,k\in\mathbb{Z}_d\\s_kX_{ij}X_{ik}\leq0}}\sqrt{|\theta_k|}\ket{i,j,k,s_k,s_kX_{ij}X_{ik}}\ket{0}\Big)\\
=\,&\frac{1}{\sqrt{N\|\theta\|_1}}\Big(\sum\limits_{\substack{i\in\mathbb{Z}_N,k\in\mathbb{Z}_d\\\theta_kX_{ij}X_{ik}>0}}\sqrt{\theta_kX_{ij}X_{ik}}\ket{i,j,k,s_k,s_kX_{ij}X_{ik}}\ket{1}\\
&+\big(\sum\limits_{\substack{i\in\mathbb{Z}_N,k\in\mathbb{Z}_d\\\theta_kX_{ij}X_{ik}>0}}\sqrt{|\theta_k|(1-s_kX_{ij}X_{ik})}\ket{i,j,k,s_k,s_kX_{ij}X_{ik}}+\sum\limits_{\substack{i\in\mathbb{Z}_N,k\in\mathbb{Z}_d\\\theta_kX_{ij}X_{ik}\leq0}}\sqrt{|\theta_k|}\ket{i,j,k,s_k,s_kX_{ij}X_{ik}}\big)\ket{0}\Big)\\
=\,&\sqrt{a_+}\ket{\phi_1}\ket{1}+\sqrt{1-a_+}\ket{\phi_0}\ket{0}, \mbox{~~~~for  }
a_+= \sum\limits_{\substack{i\in\mathbb{Z}_N,k\in\mathbb{Z}_d\\\theta_kX_{ij}X_{ik}>0}}\theta_kX_{ij}{X_{ik}}/(N\|\theta\|_1).
\end{align*}
By Theorem~\ref{thm:amplitude_estimation}, with failure probability at most $1/1000$, we can estimate $a_+$ with additive error $\beta/(8\|\theta\|_1)$ using  $\mathcal{O}(\|\theta\|_1/\beta)$ applications of $O_X$, $O^\dagger_X$, $U_\theta$, $U_\theta^\dagger$, and $\tilde{\mathcal{O}}(
\|\theta\|_1/\beta)$ elementary gates. Note that our algorithm knows $\|\theta\|_1$ because it is stored in the root of $KP_\theta$. 
We similarly estimate 
\[
a_-= -\sum\limits_{\substack{i\in\mathbb{Z}_N,k\in\mathbb{Z}_d\\\theta_kX_{ij}X_{ik}<0}} \theta_kX_{ij}{X_{ik}}/(N\|\theta\|_1)
\] 
with additive error $\beta/(8\|\theta\|_1)$.
Hence we estimate $\frac{2}{N}\sum\limits_{i\in\intg_N} X_{ij}\cdot(X\theta)_i=2\|\theta\|_1\cdot(a_+-a_-)$ with additive error $\beta/2$.

For the second term of the right-hand side of Equation~(\ref{eq:nabla_lasso}) we use a very similar strategy: we separately estimate its positive term and its negative term, each with additive error $\beta/4$, using $\mathcal{O}(1/\beta)$ applications of $O_X$, $O_X^\dagger$, $O_y$, $O_y^\dagger$, and $\tilde{\mathcal{O}}(1/\beta)$ elementary gates, respectively. Therefore, we can estimate $\frac{1}{N}\sum\limits_{i\in\intg_N} X_{ij}y_i$ with additive error $\beta/2$.

Combining the previous estimations, with failure probability at most $1/100$, we  estimate $\nabla_j L_S(\theta)$ with additive error $\beta$. Since $\norm{\theta}_1\leq 1$, we use  $\mathcal{\tilde{O}}(1/\beta)$ applications of $O_X$, $O_X^\dagger$, $O_y$, $O_y^\dagger$, $O_{KP_\theta}$, $O^\dagger_{KP_\theta}$, and elementary gates. 
By repeating the procedure $O(\log (1/\delta))$ times and taking the median of the outputs, we can decrease the failure probability from at most $1/100$ to at most $\delta$.
\end{proof}

Next we show how to estimate the value of the loss function itself at a given $\theta$:

\begin{theorem}\label{thm:error_KPtree}
Let $\theta\in B^d_1$, and $\beta, \delta >0$. Suppose we have a KP-tree $KP_{\theta}$ of vector $\theta$ and can make quantum queries to $O_{KP_{\theta}}: \ket{\ell,k}\ket{0}\rightarrow \ket{\ell,k}\ket{KP_{\theta}(\ell,k)}$. 
Then we can implement $\tilde{U}_{L_S}: \ket{0}\rightarrow \ket{\Lambda}$ such that after measuring the state~$\ket{\Lambda}$, with probability $\geq 1-\delta$ the first register $\lambda$ of the outcome will satisfy 
$|\lambda- L_S(\theta)|\leq \beta$, by using $\mathcal{\tilde{O}}(\frac{\log (1/\delta)}{\beta})$  applications of $O_X$, $O^\dagger_X$, $O_y$, $O^\dagger_y$, $O_{KP_\theta}$, $O^\dagger_{KP_\theta}$, and elementary gates.
\end{theorem}

\begin{proof}
Recall
\begin{align}
  L_S(\theta)=\frac{1}{N} \sum\limits_{i\in\intg_N} |\langle x_i,\theta\rangle-y_i|^2=\frac{1}{N}\sum\limits_{i\in\intg_N}\langle x_i,\theta\rangle^2-\frac{2}{N}\sum\limits_{i\in\intg_N}y_i\langle x_i,\theta\rangle+\frac{1}{N}\sum\limits_{i\in\intg_N} y^2_i.\label{eq:lasso}
\end{align}
We use the positive controlled rotation gate defined in the proof of Theorem~\ref{thm:nabla_error_new}. By Theorem~\ref{thm:StatePrepare}, we can implement $U_{\theta}:\ket{0}\ket{0}\rightarrow \ket{\theta}$, where $\ket{\theta}=\frac{1}{\sqrt{\|\theta\|_1}}\sum\limits_{j\in\mathbb{Z}_d} \sqrt{|\theta_j|}\ket{j}\ket{sign(\theta_j)}$, using $\mathcal{\tilde{O}}(1)$ queries to $O_{KP_\theta}$ and elementary gates.

We start by estimating (with additive error $\beta/4$) the first term on the right-hand side of Equation~(\ref{eq:lasso}), which is 
\[
\frac{1}{N}\sum\limits_{i\in\intg_N}\langle x_i,\theta\rangle^2=\frac{1}{N}\sum\limits_{i\in\intg_N}\sum\limits_{k_1,k_2 \in \mathbb{Z}_d}\theta_{k_1}\theta_{k_2}X_{ik_1}X_{ik_2}. 
\]
Let $U_u:\ket{i}\ket{k_1}\ket{s_1}\ket{k_2}\ket{s_2}\ket{0}\rightarrow \ket{i}\ket{k_1}\ket{s_1}\ket{k_2}\ket{s_2}\ket{s_1s_2X_{ik_1}X_{ik_2}}$, where the last register is the product of four numbers. This can be implemented (with negligible error) via $\mathcal{O}(1)$ queries to $O_X$, $O^\dagger_X$, and $\tilde{\mathcal{O}}(1)$ elementary gates. 

We generate  $\frac{1}{\sqrt{N}}\sum\limits_{i\in\intg_N}\ket{i}\ket{0}\ket{0}\ket{0}\ket{0}\ket{0}\ket{0}$, and apply $U_\theta$ twice to obtain $\frac{1}{\sqrt{N}}\sum\limits_{i\in\intg_N}\ket{i}\ket{\theta}^{\otimes 2}\ket{0}\ket{0}$. Applying $U_u$ on the second to seventh registers and applying $U_{CR^+}$ on the last two, we obtain
\begin{align*} &\frac{1}{\sqrt{N\|\theta\|^2_1}}\Big(\sum\limits_{\substack{i\in\mathbb{Z}_N,k_1,k_2\in\mathbb{Z}_d\\s_{k_1}s_{k_2}X_{ik_1}X_{ik_2}>0}}\sqrt{|\theta_{k_1}\theta_{k_2}|}\ket{Z_{i,k_1,k_2}}(\sqrt{s_{k_1}s_{k_2}X_{ik_1}X_{ik_2}}\ket{1}+\sqrt{1-s_{k_1}s_{k_2}X_{ik_1}X_{ik_2}}\ket{0})\\
&+\sum\limits_{\substack{i\in\mathbb{Z}_N,k_1,k_2\in\mathbb{Z}_d\\s_{k_1}s_{k_2}X_{ik_1}X_{ik_2}\leq0}}\sqrt{|\theta_{k_1}\theta_{k_2}|}\ket{Z_{i,k_1,k_2}}\ket{0}\Big)\\
=\,&\frac{1}{\sqrt{N\|\theta\|^2_1}}\Big(\sum\limits_{\substack{i\in\mathbb{Z}_N,k_1,k_2\in\mathbb{Z}_d\\\theta_{k_1}\theta_{k_2}X_{ik_1}X_{ik_2}>0}}\sqrt{\theta_{k_1}\theta_{k_2}X_{ik_1}X_{ik_2}}\ket{Z_{i,k_1,k_2}}\ket{1}\\
&+\big(\sum\limits_{\substack{i\in\mathbb{Z}_N,k_1,k_2\in\mathbb{Z}_d\\\theta_{k_1}\theta_{k_2}X_{ik_1}X_{ik_2}>0}}\sqrt{|\theta_{k_1}\theta_{k_2}|(1-s_{k_1}s_{k_2}X_{ik_1}X_{ik_2})}\ket{Z_{i,k_1,k_2}}+\sum\limits_{\substack{i\in\mathbb{Z}_N,k_1,k_2\in\mathbb{Z}_d\\\theta_{k_1}\theta_{k_2}X_{ik_1}X_{ik_2}\leq0}}\sqrt{|\theta_{k_1}\theta_{k_2}|}\ket{Z_{i,k_1,k_2}}\big)\ket{0}\Big)\\
=\,&\sqrt{a_+}\ket{\phi_1}\ket{1}+\sqrt{1-a_+}\ket{\phi_0}\ket{0},\mbox{~~~~for }a_+= \hspace*{-1.5em}\sum\limits_{\substack{i\in\mathbb{Z}_N,k_1,k_2\in\mathbb{Z}_d\\\theta_{k_1}\theta_{k_2}X_{ik_1}X_{ik_2}>0}}\hspace*{-1.5em}\frac{\theta_{k_1}\theta_{k_2}X_{ik_1}{X_{ik_2}}}{N\|\theta\|^2_1},
\end{align*}
where $s_{k_1}=sign(\theta_{k_1})$, $s_{k_2}=sign(\theta_{k_2})$, $Z_{i,k_1,k_2}=(i,k_1,s_{k_1},k_2,s_{k_2},s_{k_1}s_{k_2}X_{ik_1}X_{ik_2})$. 

By applying Theorem~\ref{thm:amplitude_estimation}, with failure probability at most $1/1000$, we can estimate $a_+$ with additive error $\beta/(6\|\theta\|^2_1)$ using  $\mathcal{O}(\|\theta\|^2_1/\beta)$ applications of $O_X$, $O^\dagger_X$, $U_\theta$, $U_\theta^\dagger$, and $\tilde{\mathcal{O}}(
\|\theta\|^2_1/\beta)$ elementary gates. Note that our algorithm knows $\|\theta\|_1$ because it is stored in the root of $KP_\theta$. Similarly we estimate
\[
a_-= -\sum\limits_{\substack{i\in\mathbb{Z}_N,k_1,k_2\in\mathbb{Z}_d\\\theta_{k_1}\theta_{k_2}X_{ik_1}X_{ik_2}<0}} \frac{\theta_{k_1}\theta_{k_2}X_{ik_1}{X_{ik_2}}}{N\|\theta\|^2_1}
\]
with the same additive error. Hence we can estimate 
\[
\frac{1}{N} \sum\limits_{i\in \mathbb{Z}_N,k_1,k_2\in\mathbb{Z}_d}\theta_{k_1}\theta_{k_2}X_{ik_1}X_{ik_2}=\|\theta\|^2_1\cdot(a_+-a_-)
\]
with additive error $\beta/3$.

For the second and third terms of the right-hand side of Equation~(\ref{eq:lasso}), we use a similar strategy to estimate each with additive error $\beta/3$, using $\mathcal{O}(1/\beta)$ applications of $O_X$, $O_X^\dagger$, $O_y$, $O_y^\dagger$, $U_\theta$, $U_\theta^\dagger$, and $\tilde{\mathcal{O}}(1/\beta)$ elementary gates.

Combining the previous estimations and the fact that we can implement $U_\theta$ by $\tilde{\mathcal{O}}(1)$ queries to $O_{KP_\theta}$, $O^\dagger_{KP_\theta}$ and $\|\theta\|_1\leq 1$, with failure probability at most $1/100$, we can estimate $L_S(\theta)$ with additive error $\beta$ by using  $\mathcal{\tilde{O}}(1/\beta)$ applications of $O_X$, $O_X^\dagger$, $O_y$, $O_y^\dagger$, $O_{KP_\theta}$, $O^\dagger_{KP_\theta}$, and elementary gates. 
By repeating the procedure $\Theta(\log (1/\delta))$ times and taking the median among the outputs, we can decrease the failure probability from at most $1/100$ to at most $\delta$.
\end{proof}

If we have multiple vectors $\theta^0,\ldots,\theta^{m-1}$, then we can apply the previous theorem conditioned on the index of the vector we care about:

\begin{corollary}\label{cor:multi_estimation}
Let $\theta^0,\theta^1,\ldots,\theta^{m-1}\in B^d_1$, and $\beta, \delta >0$. Suppose for all $h\in\mathbb{Z}_m$, we have a KP-tree $KP_{\theta^h}$ of vector $\theta^h$ and can make quantum queries to $O_{KP_{\theta}}: \ket{h,\ell,k}\ket{0}\rightarrow \ket{h,\ell,k}\ket{KP_{\theta^h}(\ell,k)}$.
Then we can implement $\tilde{U}_{L_S}: \ket{h}\ket{0}\rightarrow \ket{h}\ket{\Lambda}$ such that for all $h\in \mathbb{Z}_m$, after measuring the state~$\ket{\Lambda}$, with probability $\geq 1-\delta$ the first register $\lambda$ of the outcome will satisfy 
$|\lambda- L_S(\theta^{h})|\leq \beta$, by using $\mathcal{\tilde{O}}(\frac{\log (1/\delta)}{\beta})$  applications of $O_X$, $O^\dagger_X$, $O_y$, $O^\dagger_y$, $O_{KP_\theta}$, $O^\dagger_{KP_\theta}$, and elementary gates.
\end{corollary}

\subsection{Quantum algorithms for Lasso with respect to $S$}\label{sec:Q_Lasso_solver_worst}

In this subsection, we will show how to find an approximate minimizer for Lasso with respect to a given sample set $S$. The following algorithm simply applies the Frank-Wolfe algorithm to find an $\epsilon$-minimizer for Lasso with respect to the sample set $S$ given $C$, a guess for the curvature constant $C_{L_S}$ (which our algorithm does not know in advance). Note that to find an $s\in B^d_1$ such that $\langle s,\nabla L_S(\theta^{t})\rangle\leq \min\limits_{s'\in \mathcal{X}}\langle s',\nabla L_S(\theta^{t})\rangle +  \tau_t C_{L_S}/4$, it suffices to only check $s\in \{\pm e_0,\ldots,\pm e_{d-1}\}$ because the domain is $B^d_1$ and $\nabla L_S$ is a linear function in~$\theta$. Also, by Theorem~\ref{thm:curve_bound}, the curvature constant $C_{L_S}$ of loss function $L_S$ is at most $8$ because $(x_i,y_i)$ is in $[-1,1]^d\times [-1,1]$ for all $i\in\intg_N$. 

\begin{algorithm}[hbt]
\SetKwData{Left}{left}\SetKwData{This}{this}\SetKwData{Up}{up}
\SetKwFunction{Union}{Union}\SetKwFunction{FindCompress}{FindCompress}
\SetKwInOut{Input}{input}\SetKwInOut{Output}{output}
\Input{a positive value $C$; additive error $\epsilon$;} 
\SetAlgoLined
Let $\theta^0$ be the $d$-dimensional all-zero vector\;
Let $T= 6\cdot\lceil \frac{C}{\epsilon} \rceil$\;
  \For{$t\leftarrow 0$ \KwTo $T$}{
   $\tau_t= \frac{2}{t+2}$\;
    Let $s\in \{\pm e_0,\ldots,\pm e_{d-1}\}$ be such that $\langle\nabla L_S(\theta^t),s\rangle \leq \min\limits_{j' \in \mathbb{Z}_d}-|\nabla_{j'} L_S(\theta^t)| +\frac{C}{8t+16}$\;
  $\theta^{t+1}=(1-\tau_t) \theta^{t}+\tau_ts$\;
  }
  \Output{$\theta^T$;}
 \caption{The algorithm for Lasso with a guess~$C$ for the value of the curvature constant}
 \label{Alg:FW_LASSO_Known}
\end{algorithm}

It is worth mentioning that Algorithm~\ref{Alg:FW_LASSO_Known} also outputs an $\epsilon$-minimizer if its input $C$ equals the curvature constant $C_{L_S}$ approximately instead of exactly. For example, suppose we only know that the curvature constant $C_{L_S}$ is between $C$ and $2C$, where $C$ is the input in Algorithm~\ref{Alg:FW_LASSO_Known}. Then the output of Algorithm~\ref{Alg:FW_LASSO_Known} is still an $\epsilon$-minimizer. We can see this by first observing that the error we are allowed to make for the linear subproblem in iteration $t$ is $\frac{C_{L_S}}{4t+8}\geq \frac{C}{8t+16}$, and hence by Theorem~\ref{thm:FW}, after $T=6\cdot\lceil \frac{C}{\epsilon}\rceil$ iterations, the output $\theta^T$ is a $\frac{3C}{(T+2)}=\frac{3C}{6\cdot\lceil C/\epsilon \rceil +2}$-minimizer for~$L_S$. Because $\frac{3C}{6\cdot\lceil C/\epsilon \rceil +2}\leq \epsilon$, the output $\theta^T$ is therefore an $\epsilon$-minimizer.

In the Lasso case, we do not know how to find a positive number $C$ such that $C_{L_S}\in [C,2C]$, but we know $C_{L_S}\leq 8$ by Theorem~\ref{thm:curve_bound}. Hence we can try different intervals of possible values for $C_{L_S}$: we apply Algorithm~\ref{Alg:FW_LASSO_Known} with different input $C=8,4,2,1,1/2,\ldots, 2^{-\lceil\log (1/\epsilon)\rceil}$, and then we collect all outputs of Algorithm~\ref{Alg:FW_LASSO_Known} with those different inputs, as candidates. After that, we compute the objective values of all those candidates, and output the one with minimum objective value. If $C_{L_S}\in (\epsilon,8]$, then at least one of the values we tried for $C$ will be within a factor of~2 of the actual curvature constant $C_{L_S}$. Hence one of our candidates is an $\epsilon$-minimizer.

However, we also need to deal with the case that $C_{L_S}\leq \epsilon$. In this case, we consider the ``one-step'' version of the Frank-Wolfe algorithm, where the number of iterations is~1. But now we do not estimate $\langle\nabla L_S(\theta^t),s\rangle$ anymore (i.e., we do not solve linear subproblems anymore). We find that the only possible directions are the vertices of the $\ell_1$-ball, and $\theta^0$ is the all-zero vector, implying that $\theta^1$, the output of one-step Frank-Wolfe, must be in $I=\{\pm e_0/3,\ldots,\pm e_{d-1}/3\}$ by the update rule of Frank-Wolfe. Besides, $C_{L_S}\leq \epsilon$ implies that $\theta^1$ is a $\frac{3C_{L_S}}{1+2}\leq \epsilon$-minimizer for Lasso. Hence we simply output a $v=\arg\min\limits_{v'\in I}L_S(v')$ if $C_{L_S}\leq\epsilon$.

Combining the above arguments gives the following algorithm:

\begin{algorithm}[hbt]
\SetKwData{Left}{left}\SetKwData{This}{this}\SetKwData{Up}{up}
\SetKwFunction{Union}{Union}\SetKwFunction{FindCompress}{FindCompress}
\SetKwInOut{Input}{input}\SetKwInOut{Output}{output}
\Input{$\epsilon$;} 
Let $v\in \{\pm e_0/3,\ldots,\pm e_{d-1}/3\}$ be such that $ L_S(v)-\min_{j\in \mathbb{Z}_d} L_S(\pm e_j/3)\leq \epsilon/10$\;
Let candidate set $A=\{v\}$\;
\SetAlgoLined
      \For{$C\leftarrow 8,4,2,1,\frac{1}{2}, \ldots,2^{-\lceil\log (1/\epsilon)\rceil-1}$}{
  RUN Algorithm~\ref{Alg:FW_LASSO_Known} with inputs $C$ and $\epsilon/10$\;
  ADD the output of Algorithm~\ref{Alg:FW_LASSO_Known} to $A$\;
  }
  \Output {$\arg\min_{w\in A} L_S(w)$;}
 \caption{The algorithm for Lasso}
 \label{Alg:Q_Lasso}
\end{algorithm}

\begin{theorem}\label{thm:Q_lasso_empirical}
Let $S=\{(x_i,y_i)\}_{i=0}^{N-1}$ be the given sample set stored in QROM. %Suppose we can make quantum queries to the oracle
%$O_X:\ket{i}\ket{j}\ket{0}\rightarrow \ket{i}\ket{j}\ket{X_{ij}}$ and the oracle $O_y:\ket{i}\ket{0}\rightarrow \ket{i}\ket{y_{i}}$.
For each $\epsilon \in (0,0.5)$, there exists a bounded-error quantum algorithm that finds an $\epsilon$-minimizer for Lasso with respect to sample set $S$ using $\tilde{\mathcal{O}}(\frac{\sqrt{d}}{\epsilon^{2}})$ time
%queries to $O_X$, $O_X^\dagger$, $O_y$, $O_y^\dagger$, and elementary gates, 
%$\poly\log (\frac{1}{\epsilon},d)$ qubits, 
and $\tilde{\mathcal{O}}(\frac{1}{\epsilon})$ QRAM and classical space.
\end{theorem}

\begin{proof}
We will implement Algorithm~\ref{Alg:Q_Lasso} in $\tilde{\mathcal{O}}(\frac{\sqrt{d}}{\epsilon^{2}})$ time
%, $\poly\log (\frac{1}{\epsilon},d)$ qubits, 
and $\tilde{\mathcal{O}}(\frac{1}{\epsilon})$  QRAM space. Below we analyze its different components. 

\paragraph{Analysis of Algorithm~\ref{Alg:FW_LASSO_Known}.}
We first show that we can implement Algorithm~\ref{Alg:FW_LASSO_Known} in $\tilde{\mathcal{O}}(\frac{\sqrt{d}}{\epsilon^{2}})$ time.
%and $\tilde{\mathcal{O}}(\frac{1}{\epsilon})$  QRAM and classical space. 
%There are $\tilde{\mathcal{O}}(1/\epsilon)$ iterations, so it suffices to show that the slowest iteration in Algorithm~\ref{Alg:FW_LASSO_Known} uses $\tilde{\mathcal{O}}(\frac{\sqrt{d}}{\epsilon})$ time. 
Because $C_{L_S}\leq 8$ (Theorem~\ref{thm:curve_bound}), the number of iterations for Algorithm~\ref{Alg:FW_LASSO_Known} with input $C=C_{L_S}$ is at most $6\cdot\lceil \frac{8}{\epsilon}\rceil$. However, as we mentioned above, we don't know how large $C_{L_S}$ is exactly, so we try all possible inputs (of Algorithm~\ref{Alg:FW_LASSO_Known}) in Algorithm~\ref{Alg:Q_Lasso}. Note that for every input $C\in\{8,4,2,1,\frac{1}{2}, \ldots,2^{-\lceil\log (1/\epsilon)\rceil-1}\}$ and for every number of iterations $t\in \{1,\ldots,6\cdot \lceil\frac{C}{\epsilon}\rceil\}$, $\frac{C}{4t+8}$ is at least $\frac{\epsilon}{10}$, so it suffices to ensure that in each iteration in each of our runs of Algorithm~\ref{Alg:FW_LASSO_Known}, the additive error for the approximate linear subproblem is $\leq\frac{\epsilon}{10}$.

Suppose we have $KP_{\theta^t}$ for each iteration $t$ of Algorithm~\ref{Alg:FW_LASSO_Known}, and suppose we can make queries to $O_{KP_{\theta^t}}$, then by Theorem~\ref{thm:nabla_error_new}, one can implement $\tilde{U}_{\nabla L_S}: \ket{j}\ket{0}\rightarrow \ket{j}\ket{\Lambda}$ such that for all $j\in \mathbb{Z}_d$, after measuring the state $\ket{\Lambda}$, with probability $\geq 1-\frac{
\epsilon^2}{2d\cdot 10^{20}\cdot \log^6(1/\epsilon)}$ the first register $\lambda$ of the measurement outcome will satisfy
$|\lambda-\nabla_j L_S(\theta)|\leq \frac{\epsilon}{20},$ by using $\mathcal{\tilde{O}}(\frac{\log (d/\epsilon)}{\epsilon})$
%applications of $O_X$, $O^\dagger_X$, $O_y$, $O^\dagger_y$, 
time and queries to $O_{KP_{\theta^t}}$, $O^\dagger_{KP_{\theta^t}}$. Then by Theorem~\ref{thm:min_finding_approx}, with failure probability at most %$\frac{\epsilon}{2000\log(1/\epsilon)}+1000\log(\frac{2000\log(1/\epsilon)}{\epsilon})\sqrt{2d\cdot\frac{\epsilon^2}{2d\cdot 10^{12}\cdot\log^6(1/\epsilon)}}\leq$
$\frac{\epsilon}{10000\log(1/\epsilon)}$, one can find $s\in \{\pm e_0,\ldots,\pm e_{d-1}\}$ such that $\langle\nabla L_S(\theta^t),s\rangle \leq \min\limits_{j' \in \mathbb{Z}_d}-|\nabla_{j'} L_S(\theta^t)| +2\cdot\frac{\epsilon}{20}$, by using $\mathcal{\tilde{O}}(\sqrt{d}\cdot \log(1/\epsilon))$ applications of $\tilde{U}_{\nabla L_S}$ and $\tilde{U}^\dagger_{\nabla L_S}$, and $\tilde{\mathcal{O}}(\sqrt{d})$ elementary gates. %Repeating the procedure for $\Theta(\log\log (\frac{1}{\epsilon})\cdot\log (\frac{1}{\epsilon}))$ times and taking the median among them, we can decrease the failure probability from at most $\frac{1}{1000}$ to at most $\frac{\epsilon}{1000\log(1/\epsilon)}$. \ynote{It's really weird to say take median among them.}

For each iteration $t$ in Algorithm~\ref{Alg:FW_LASSO_Known}, we also maintain $KP_{\theta^t}$ and hence we can make quantum queries to $O_{KP_{\theta^t}}$. The cost for constructing $KP_{\theta^0}$ and the cost for updating $KP_{\theta^t}$ to $KP_{\theta^{t+1}}$ is $\tilde{\mathcal{O}}(1)$ for both time and space by Theorem~\ref{thm:KP_update_cost}. Moreover, the total number of iterations $T$ is at most $6\cdot\lceil\frac{8}{\epsilon}\rceil$ in Algorithm~\ref{Alg:FW_LASSO_Known} because $C_{L_S}\leq 8$, and hence the space cost for maintaining $KP_{\theta^t}$ and implementing $O_{KP_{\theta^t}}$ is $\tilde{\mathcal{O}}(\frac{1}{\epsilon})$ bits. Hence we can implement Algorithm~\ref{Alg:FW_LASSO_Known} with failure probability at most $\lceil\frac{8}{\epsilon}\rceil\cdot\frac{6\epsilon}{10000\log(1/\epsilon)}$ using $\tilde{\mathcal{O}}(\frac{\sqrt{d}}{\epsilon^{2}})$ time  %queries to $O_X$, $O^\dagger_X$, $O_y$, $O^\dagger_y$, elementary gates 
and $\tilde{\mathcal{O}}(\frac{1}{\epsilon})$ bits of QRAM and classical space.

\paragraph{Analysis of Algorithm~\ref{Alg:Q_Lasso}.}
Now we show how to implement Algorithm~\ref{Alg:Q_Lasso} with failure probability at most $1/10$ using $\tilde{\mathcal{O}}(\frac{\sqrt{d}}{\epsilon^2})$ time.
%queries to $O_X$, $O_y$, and elementary gates. 
%Suppose we have $U_{L_S}: \ket{j}\ket{0}\rightarrow \ket{j}\ket{L_j}$ such that for all $j\in \mathbb{Z}_d$, $L_j=L_S(e_j/3)$. By using quantum minimum-finding algorithm in Theorem~\ref{thm:min_finding_approx}, with failure probability at most $\frac{1}{1000}$, one can find $v\in \{\pm e_0/3,\ldots,\pm e_{d-1}/3\}$ such that $ L_S(v)=\min_{j\in \mathbb{Z}_d} L_S(\pm e_j/3)$ using $\mathcal{O}(\sqrt{d})$ queries to $U_{L_S}$ and $U^\dagger_{L_S}$. Also, 
By Corollary~\ref{cor:multi_estimation}, one can implement $\tilde{U}_{L_S}: \ket{j}\ket{0}\rightarrow \ket{j}\ket{\Lambda}$ such that for all $j\in \mathbb{Z}_d$, after measuring the state $\ket{\Lambda}$, with failure probability at most $\frac{1}{2d\cdot10^{16}}$ the first register $\lambda$ of the outcome will satisfy 
$|\lambda- L_S(e_j/3)|\leq \epsilon/20$ using $\mathcal{\tilde{O}}(\frac{1}{\epsilon})$ time.
%applications of $O_X$, $O^\dagger_X$, $O_y$, $O^\dagger_y$,and $\tilde{\mathcal{O}}(\frac{1}{\epsilon})$ elementary gates. 
Then by Theorem~\ref{thm:min_finding_approx}, with failure probability at most $0.0001+1000\cdot\log(1000)\sqrt{\frac{2d}{2d\cdot10^{16}}}\leq \frac{2}{1000}$ we can find $v\in \{\pm e_0/3,\ldots,\pm e_{d-1}/3\}$ such that $ L_S(v)-\min_{j\in \mathbb{Z}_d} L_S(\pm e_j/3)\leq 2\cdot\epsilon/20=\epsilon/10$ by using $\mathcal{\tilde{{O}}}(\sqrt{d})$ applications of $\tilde{U}_{L_S}$ and $\tilde{U}_{L_S}^\dagger$ and $\tilde{\mathcal{{O}}}(\sqrt{d})$  elementary gates, and hence  $\mathcal{\tilde{O}}(\frac{\sqrt{d}}{\epsilon})$ time.
%applications of $O_X$, $O^\dagger_X$, $O_y$, $O^\dagger_y$,and $\tilde{\mathcal{O}}(\frac{\sqrt{d}}{\epsilon})$ elementary gates. 

Because Algorithm~\ref{Alg:Q_Lasso} runs Algorithm~\ref{Alg:FW_LASSO_Known} $\lceil \log (1/\epsilon)\rceil$ times and each run fails with probability at most $\lceil\frac{8}{\epsilon}\rceil\cdot\frac{6\epsilon}{10000\log(1/\epsilon)}$, the candidate set $A$, with failure probability $\lceil\frac{8}{\epsilon}\rceil\cdot\frac{6\epsilon}{10000\log(1/\epsilon)}\cdot \lceil \log (1/\epsilon)\rceil+\frac{2}{1000}\leq\frac{1}{20}$, contains an $\frac{\epsilon}{10}$-minimizer. To output $\arg\min_{w\in A} L_S(w)$, we use Theorem~\ref{thm:error_KPtree} to evaluate $L_S(w)$ for all $w\in A$ with additive error $\frac{\epsilon}{10}$ with failure probability at most $\frac{1}{40\log(1/\epsilon)}$, and hence we find an $\epsilon/10$-minimizer among $A$ with probability at least $1-1/20-\lceil \log (1/\epsilon)\rceil\cdot \frac{1}{40\log(1/\epsilon)}\geq 0.9$. Because the candidate set $A$ contains an $\frac{\epsilon}{10}$-minimizer for Lasso, the $\frac{\epsilon}{10}$-minimizer among $A$ is therefore an $\epsilon$-minimizer for Lasso. The QRAM and classical space cost for each run is at most $\tilde{\mathcal{O}}(\frac{1}{\epsilon})$ because the space cost for Algorithm~\ref{Alg:FW_LASSO_Known} is $\tilde{\mathcal{O}}(\frac{1}{\epsilon})$. Hence the total cost for implementing Algorithm~\ref{Alg:Q_Lasso} is $\tilde{\mathcal{O}}(\frac{\sqrt{d}}{\epsilon^{2}})$ time and $\tilde{\mathcal{O}}(\frac{1}{\epsilon})$ bits of QRAM and classical space.
\end{proof}

\subsection{Quantum algorithms for Lasso with respect to $\mathcal{D}$}

In the previous subsection, we showed that we can find an $\epsilon$-minimizer for Lasso with respect to sample set~$S$. Here we show how we can find an $\epsilon$-minimizer for Lasso with respect to distribution~$\mathcal{D}$. 
First sample a set $S$ of $N=\tilde{\mathcal{O}}((\log d)/\epsilon^2)$ i.i.d.\ samples from $\mathcal{D}$, which is the input that will be stored in QROM, and then find an $\epsilon/2$-minimizer for Lasso with respect to $S$ by Theorem~\ref{thm:Q_lasso_empirical}. By Theorem~\ref{thm:Rademacher_Lasso}, with high probability, an $\epsilon/2$-minimizer for Lasso with respect to $S$ will be an $\epsilon$-minimizer for Lasso with respect to distribution $\mathcal{D}$. Hence we obtain the following corollary:

\begin{corollary} Let $S=\{(x_i,y_i)\}_{i=0}^{N-1}$ be the given sample set, sampled i.i.d.\ from $\mathcal{D}$. 
%Suppose we can make quantum queries to the oracle $O_X:\ket{i}\ket{j}\ket{0}\rightarrow \ket{i}\ket{j}\ket{X_{ij}}$ and the oracle $O_y:\ket{i}\ket{0}\rightarrow \ket{i}\ket{y_{i}}$.
For arbitrary $\epsilon >0$, if $N=\tilde{\mathcal{O}}(\frac{\log d}{\epsilon^2})$, then there exists a bounded-error quantum algorithm that finds an $\epsilon$-minimizer for Lasso with respect to distribution $\mathcal{D}$ using $\tilde{\mathcal{O}}(\frac{\sqrt{d}}{\epsilon^{2}})$ queries to $O_X$, $O_y$ and elementary gates, and using $\tilde{\mathcal{O}}(\frac{1}{\epsilon})$ space (QRAM and classical bits).
\end{corollary}

We can also use Theorem~\ref{thm:QRAM_free} to avoid the usage of QRAM in the above corollary with $\mathcal{\tilde{O}}(1/\epsilon)$ extra overhead.
\begin{corollary} Let $S=\{(x_i,y_i)\}_{i=0}^{N-1}$ be the given sample set, sampled i.i.d.\ from $\mathcal{D}$.
For arbitrary $\epsilon >0$, if $N=\tilde{\mathcal{O}}(\frac{\log d}{\epsilon^2})$, then there exists a bounded-error quantum algorithm that finds an $\epsilon$-minimizer for Lasso with respect to distribution $\mathcal{D}$ using $\tilde{\mathcal{O}}(\frac{\sqrt{d}}{\epsilon^{3}})$ queries to $O_X$, $O_y$ and elementary gates, and using $\tilde{\mathcal{O}}(\frac{1}{\epsilon})$ classical bits.
\end{corollary}

\section{Quantum query lower bounds for Lasso}\label{sec:q_lower}
In this section we prove a quantum lower bound of $\Omega(\sqrt{d}/\epsilon^{1.5})$ queries for Lasso. To show such a lower bound, we define a certain set-finding problem, and show how it can be solved by an algorithm for Lasso. After that, we show that the worst-case set-finding problem can be seen as the composition of two problems, which have query complexities $\Omega(\sqrt{d/\epsilon})$ and $\Omega(1/\epsilon)$, respectively. Then the composition property of the quantum adversary bound implies a $\Omega(\sqrt{d/\epsilon}\cdot 1/\epsilon)=\Omega(\sqrt{d}/\epsilon^{1.5})$ query lower bound for Lasso. 

\subsection{Finding a hidden set $W$ using a Lasso solver}\label{sec:distributed_HSF_to_Lasso}

In this subsection we define the \emph{distributional set-finding problem}, and show how to reduce this to Lasso.
Let $p\in (0,1/2)$, $W\subset \mathbb{Z}_d$, and $\overline{W}=\mathbb{Z}_d\setminus W$. Define the distribution $\mathcal{D}_{p,W}$ over $(x,y)\in\{-1,1\}^d\times \{-1,1\}$ as follows. For each $j'\in \overline{W}$, $x_{j'}$ is generated according to  $\Pr[x_{j'}=1]=\Pr[x_{j'}=-1]=1/2$, and for each $j\in W$, $x_j$ is generated according to  $\Pr[x_{j}=1]=1/2+p$. And $y$ is generated according to $\Pr[y=1]=1$. The goal of the distributional set-finding problem $\mbox{DSF}_{\mathcal{D}_{p,W}}$ with respect to $\mathcal{D}_{p,W}$ is to output a set $\tilde{W}$ such that $|\tilde{W}\Delta W|\leq w/200$, given $M$ samples from~$\mathcal{D}_{p,W}$. One can think of the $M\times d$ matrix of samples as ``hiding'' the set $W$: the columns corresponding to $j\in W$ are likely to have more 1s than $-1$s, while the columns corresponding to $j\in \overline{W}$ have roughly as many 1s as $-1$s.

We first show some basic properties of $L_{\mathcal{D}_{p,W}}$.  In this subsection, let 
\[
\theta^*=v\cdot e_W, \mbox{ where }v=2p/(1+4p^2(w-1)) \mbox{ and }e_W=\sum\limits_{j\in W}e_j.
\]

\begin{theorem}\label{thm:properties_distinguisher}
Let $\theta$ be a vector in $\mathbb{R}^d$. We have
\begin{itemize}
    \item $L_{\mathcal{D}_{p,W}}(\theta)=\sum\limits_{j'\in\overline{W}} \theta_{j'}^2+\sum\limits_{j\in W} (\theta_j-2p)^2+4p^2\sum\limits_{j_1\in W}\sum\limits_{j_2\in W\setminus \{j_1\}}\theta_{j_1}\theta_{j_2}-4p^2w+1 $, where $w=|W|$.
    \item $\nabla_{j} L_{\mathcal{D}_{p,W}} (\theta) =\begin{dcases}
    2\theta_j-4p+8p^2\sum\limits_{\ell\in W\setminus \{j\}}\theta_{\ell}, & \text{if } j\in W,\\
    2\theta_j, &\text{otherwise}.
\end{dcases}$
    \item $[\nabla^2L_{\mathcal{D}_{p,W}}(\theta)]_{jk}=
    \begin{dcases}
    2\delta_{jk}, & \text{if } j,k\in \overline{W},\\
    (2-8p^2)\delta_{jk}+8p^2, & \text{if } j,k\in  W,\\
    0, &\text{otherwise},
\end{dcases}$
\end{itemize}
If we rearrange the order of indices such that $\{0,1,\ldots,w-1\}\in W$ and $\{w,w+1,\ldots, d-1\}\in \overline{W}$, then the Hessian of $L_{\mathcal{D}_{p,W}}$ is
$$
\nabla^2L_{\mathcal{D}_{p,W}}= \Big(\begin{matrix}
  8p^2 J_{w}+(2-8p^2)I_{w} & 0_{w\times \bar{w}}\\
  0_{\bar{w}\times w} & 2I_{\bar{w}},
\end{matrix}\Big)
$$
where $\bar{w}=d-w$. Note that this is independent of $\theta$, and $\nabla^2L_{\mathcal{D}_{p,W}} \succeq (2-8p^2) I_d$.
The unique global minimizer of $L_{\mathcal{D}_{p,W}}$ is $\theta^*$.
\end{theorem}

\begin{proof}
Note that for all distinct $j_1,j_2 \in W$,
\begin{align*}
\E_{(x,y)\sim {\mathcal{D}_{p,W}}}[x_{j_1}x_{j_2}] &=\Pr_{(x,y)\sim {\mathcal{D}_{p,W}}}[x_{j_1}=x_{j_2}]\cdot 1 + \Pr_{(x,y)\sim {\mathcal{D}_{p,W}}}[x_{j_1}=-x_{j_2}]\cdot (-1)\\
&= ((1/2+p)^2+(1/2-p)^2)-2(1/2+p)(1/2-p)=4p^2.
\end{align*}
For all distinct $j, j'$ such that $j\in \mathbb{Z}_d$ but $j'\in \overline{W}$, we have $\E_{(x,y)\sim {\mathcal{D}_{p,W}}}[x_{j}x_{j'}]=0$.

We first prove the first bullet. By definition and the facts we mentioned above, 
\begin{align*}
    L_{\mathcal{D}_{p,W}}(\theta)& = \E_{(x,y)\sim {\mathcal{D}_{p,W}}} [(\langle x,\theta \rangle -y)^2]=\E_{(x,y)\sim {\mathcal{D}_{p,W}}}[(\sum\limits_{j\in \mathbb{Z}_d}x_j\theta_j-y)^2]\\
    &=  \E_{(x,y)\sim {\mathcal{D}_{p,W}}}\left[\sum\limits_{j\in\mathbb{Z}_d} x_j^2\theta_j^2+\sum\limits_{j_1\in\mathbb{Z}_d}\sum\limits_{j_2\in \mathbb{Z}_d\setminus \{j_1\}}\theta_{j_1}\theta_{j_2}x_{j_1}x_{j_2}-2\sum\limits_{j\in \mathbb{Z}_d}x_j\theta_jy \, +y^2\right]\\
    &=\sum\limits_{j\in\mathbb{Z}_d} \theta_j^2+4p^2\sum\limits_{j_1\in W}\sum\limits_{j_2\in W\setminus \{j_1\}} \theta_{j_1}\theta_{j_2}-4p\sum\limits_{j\in W}\theta_j \, +1\\
    &= \sum\limits_{j'\in\overline{W}} \theta_{j'}^2+\sum\limits_{j\in W} (\theta_j-2p)^2+4p^2\sum\limits_{j_1\in W}\sum\limits_{j_2\in W\setminus \{j_1\}}\theta_{j_1}\theta_{j_2}-4p^2\cdot|W|+1.
\end{align*}
The last equation holds because $\sum\limits_{j\in{W}} \theta_{j}^2-4p\sum\limits_{j\in W}\theta_j =\sum\limits_{j\in W} (\theta_j-2p)^2-4p^2\cdot|W|$ by completing the square. The second and third bullets are easy to see by taking first and second partial derivatives of the expression of the first bullet.
To see $\nabla^2L_{\mathcal{D}_{p,W}} \succeq (2-8p^2) I_d$, note that $\nabla^2L_{\mathcal{D}_{p,W}} - (2-8p^2) I_d$ is block-diagonal, where the $W\times W$ block is $4p^2$ times the all-1 matrix (which is positive semidefinite) and the $\overline{W}\times\overline{W}$ block is diagonal with diagonal entries~$8p^2$ (which is positive definite). 

The minimizer $\theta^*$ is a solution of the linear system one gets by setting all derivatives of the second bullet to~0.
Because the Hessian is positive definite (we assumed $p<1/2$, hence $2-8p^2>0$), this solution is the unique minimizer.
\end{proof}

The following theorem relates the entries of an approximate minimizer~$\theta$ for Lasso with respect to distribution $\mathcal{D}_{p,W}$ to the elements of the hidden set $W$.

\begin{theorem}\label{thm:weight_distinguisher}
Let $\epsilon\in (2/d, 1/100)$, $w$ be either $\lfloor 1/\epsilon\rfloor$ or $\lfloor 1/\epsilon\rfloor-1$, $p=1/(2\lfloor 1/\epsilon\rfloor)$, and $W \subset \mathbb{Z}_d$ be a set of size $w$. For every $\theta \in B^d_1$ satisfying $L_{\mathcal{D}_{p,W}}(\theta)- L_{\mathcal{D}_{p,W}}(\theta^*)\leq \epsilon/8000$, we have   
\[
\sum\limits_{j\in W} (\theta_j-v)^2 \leq \epsilon/6400\mbox{ and }\sum\limits_{j'\in \overline{W}} \theta_{j'}^2 \leq \epsilon/6400.
\]
\end{theorem}

\begin{proof}
Let $\theta^*=v\cdot e_W$ be the minimizer of $L_{\mathcal{D}_{p,W}}$ from Theorem~\ref{thm:properties_distinguisher}.
Because $\nabla_j L_{\mathcal{D}_{p,W}}(\theta^*)=0$ for every $j\in \intg_d$ and $\nabla^2 L_{\mathcal{D}_{p,W}}$ is a constant matrix, independent of $\theta$, we have that
\begin{align*}
\epsilon/8000&\geq L_{\mathcal{D}_{p,W}}(\theta)-L_{\mathcal{D}_{p,W}}(\theta^*)\\
&= \inProd{\nabla L_{\mathcal{D}_{p,W}}(\theta^*)}{\theta-\theta^*}+ \inProd{\nabla^2 L_{\mathcal{D}_{p,W}}\cdot(\theta-\theta^*)}{\theta-\theta^*}/2 \\
&\geq 0 + (2-8p^2)\|\theta-\theta^*\|^2_2/2,
\end{align*}
which implies 
$\|\theta-\theta^*\|^2_2\leq \epsilon/(8000\cdot(1-4p^2))\leq \epsilon/6400$ from the fact that $p \leq 1/200$. Because $\theta^*=v\cdot e_W$, we have $\|\theta-\theta^*\|^2_2=\sum\limits_{j\in W} (\theta_j-v)^2+\sum\limits_{j'\in  \overline{W}} \theta_{j'}^2\leq \epsilon/6400$, and therefore
\[
\sum\limits_{j\in W} (\theta_j-v)^2 \leq \epsilon/6400\,\mbox{ and }\sum\limits_{j'\in \overline{W}} \theta_{j'}^2 \leq \epsilon/6400.
\]
\end{proof}

Note that if $\epsilon\in (2/d, 1/100)$, $1\leq w\leq\lfloor 1/\epsilon\rfloor$, $p=1/(2\lfloor 1/\epsilon\rfloor)$, then
\[
\|\theta^*\|_1=vw\leq\frac{\frac{1}{\lfloor 1/\epsilon\rfloor}\cdot \lfloor \frac{1}{\epsilon}\rfloor}{1+\frac{1}{\lfloor 1/\epsilon\rfloor^2}(1-1)}\leq 1,
\]
implying that $\theta^*\in B^d_1$, so the global minimizer actually satisfies Lasso's norm constraint. 
Now we are ready to show that algorithms for Lasso also find a good approximation to the hidden set~$W$.

\begin{theorem}\label{thm:average_HSF_to_Lasso}
Let $\epsilon\in (2/d, 1/100)$, $w$ be eitehr $\lfloor 1/\epsilon\rfloor$ or $\lfloor 1/\epsilon\rfloor-1$, $p=1/(2\lfloor 1/\epsilon\rfloor)$, and $W \subset \mathbb{Z}_d$ be a set of size $w$. Let $\theta$ be an $\epsilon/8000$-minimizer for Lasso with respect to $\mathcal{D}_{p,W}$. Then the set $\tilde{W}$ that contains the indices of the entries of $\theta$ whose absolute value is $\geq\epsilon/3$ satisfies $|W\Delta \tilde{W}|\leq w/200$. 
\end{theorem}

\begin{proof}
Because $\theta$ is an $\epsilon/8000$-minimizer for Lasso, Theorem~\ref{thm:weight_distinguisher} implies $\sum\limits_{j\in W} (\theta_j-v)^2 \leq \epsilon/6400$ and $\sum\limits_{j'\in \overline{W}} \theta_{j'}^2 \leq \epsilon/6400$. Hence 
\begin{itemize}
    \item at most $w/400$ many $j'\in \overline{W}$ have $|\theta_{j'}| \geq \sqrt{400\epsilon/(6400w)}\geq\sqrt{\epsilon/ \lfloor 1/\epsilon\rfloor }/4$,
    \item at least $w-w/400$ many $j$ $\in W$ have $\theta_j \geq v-\sqrt{400\epsilon/(6400w)}\geq v-\sqrt{\epsilon/(\lfloor 1/\epsilon\rfloor-1) }/4$.
\end{itemize}
Note that $v-\sqrt{\epsilon/ (\lfloor 1/\epsilon\rfloor-1) }/4 \geq \epsilon/3\geq \sqrt{\epsilon/ \lfloor 1/\epsilon\rfloor }/4$ for both the cases that $w=\lfloor 1/\epsilon\rfloor$ and $w=\lfloor 1/\epsilon\rfloor-1$. Hence the set $\tilde{W}$ that contains the indices of the entries whose absolute value is $\geq\epsilon/3$, omits at most
$w/400$ of the $j\in W$ and includes at most $w/400$ of the $j'$ $\in \overline{W}$.
Therefore $|W\Delta \tilde{W}|\leq w/400+w/400=w/200$.
\end{proof}
This implies that algorithms that find an $\epsilon/8000$-minimizer for Lasso with respect to $\mathcal{D}_{p,W}$ can also find a set $\tilde{W}\subset \mathbb{Z}_d$ such that $|W\Delta \tilde{W}|\leq w/200$. 

\subsection{Worst-case quantum query lower bound for the set-finding problem}\label{sec:q_lower_decomposed}
Here we will define the worst-case set-finding problem and then provide a quantum query lower bound for it. Before we step into the query lower bound for the worst-case set-finding problem, we have to introduce the adversary method and the lower bounds for two problems first. 

\begin{theorem}[\cite{Amb00}, modified Theorem~6]\label{thm:adv_bound}
Let $f(x_0,\ldots, x_{d-1} )$ be a function of $d$ inputs with values from some finite set, $\epsilon \in (0,1/2)$, and $\mathcal{X}$, $\mathcal{Y}$ be two sets of valid inputs for $f$. Let $\mathcal{R} \subseteq  \mathcal{X} \times \mathcal{Y}$ be a relation such that
\begin{itemize}
    \item For every $(x,y)\in \mathcal{R}$, $f(x)\neq f(y)$.
    \item For every $x \in \mathcal{X}$, there exist at least $m$ different $y \in \mathcal{Y}$ such that $(x, y) \in \mathcal{R}$.
    \item For every $y \in\mathcal{Y}$, there exist at least $m'$ different $x \in \mathcal{X}$ such that $(x, y) \in \mathcal{R}$.
\end{itemize}
Let $\ell_{x,j}$ be the number of $y \in \cY$ such that $(x, y) \in \cR$, and $x_j \neq y_j$ and $\ell_{y,j}$ be the number of $x \in \cX$ such that $(x, y) \in \cR$ and $x_j \neq y_j$. Let $\ell_{max}$ be the maximum of $\ell_{x,j}\cdot \ell_{y,j}$ over all $(x, y) \in \cR$ and $j \in \intg_d$ such that $x_j \neq y_j$. Then, every quantum algorithm that computes $f$ with success probability $1-\epsilon$ uses at least $(1/2-\sqrt{\epsilon(1-\epsilon)})\cdot\sqrt{ mm'/\ell_{max}}$ queries.
\end{theorem}

Using this adversary bound, we can give a query lower bound for the \emph{exact set-finding problem}: 
given input $x=x_0\ldots x_{d-1}\in\{0,1\}^d$ with at most $w$ 1s, find the set $W$ of all indices $j$ with $x_j=1$ (equivalently, learn $x$). To see the query lower bound for this problem, we consider the identity function where both domain and codomain are $\cZ=\{z\in \{0,1\}^d:|z|=w\}$, and give a lower bound for computing this. If we can compute the identity function, then we can simply check the output string $x_0,x_1,\ldots,x_{d-1}$ and collect all indices $j$ with $x_j=1$.

\begin{theorem}\label{thm:LB_multi_search_exact}
Let $w$ be an integer satisfying $0<w\leq d/2$, $W\subset \intg_d$ with size $w$, and $x\in\{0,1\}^d$ such that $x_j=1$ if $j\in W$ and $x_{j'}=0$ if $j'\in\overline{W}$. Suppose we have query access to $x$. Then every quantum bounded-error algorithm to find~$W$ makes at least $\frac{1}{8}\sqrt{dw}$ queries.
\end{theorem}

\begin{proof}
Note that if we can compute $W$, then we can compute the identity function $f:\cZ\rightarrow \cZ$.
Let $\cX=\cY=\cZ$, and consider the relation $\cR\subset \cX\times\cY$ such that for every $(x,y)\in R$,  $d_H(x,y)=2$ (which implies $f(x) \neq f(y)$). Then we know
\begin{itemize}
    \item For every $x \in \cX$, there exist at least $m=w\cdot (d-w)$ different $y \in \cY$ such that $(x, y) \in R$.
    \item For every $y \in \cY$, there exist at least $m'=w\cdot (d-w)$ different $x \in \cX$ such that $(x, y) \in R$.
    \item  $\ell_{max}=\max\limits_{\substack{(x, y) \in \cR, j \in\intg_d\\ s.t. x_j \neq y_j}}\ell_{x,j}\cdot \ell_{y,j}=w\cdot (d-w)$.
\end{itemize}
By Theorem~\ref{thm:adv_bound}, every quantum algorithm that computes $f$ with probability $\geq 9/10$ uses at least
$$
(1/2-\sqrt{1/10\cdot 9/10})\cdot \sqrt{mm'/\ell_{max}}\geq\frac{1}{5} \sqrt{w\cdot(d-w)}\geq \frac{1}{8}\sqrt{dw}
$$ 
queries.
\end{proof}

%Now let us consider the \emph{approximate set-finding problem} $\mbox{ASF}_{d,w}$: suppose there are totally $d$ elements $\{x_0,\ldots,x_{d-1}\}$ in $\{0,1\}$\rnote{just write $x\in\{0,1\}^d$} and $w$ of them are $1$, and let $W=\{j\in\intg_d: x_j=1\}$. The goal of the approximate set-finding problem is to find a set $\tilde{W}\subset\intg_d$ of indices such that $|W\Delta\tilde{W}|\leq w/1000$. Note that the approximate set-finding problem now is a relational problem: for every input there are many different possible correct outputs.
We now prove a lower bound for the \emph{approximate} set-finding problem $\mbox{ASF}_{d,w}$, which is to find a set $\tilde{W}\subset\intg_d$ such that $|W\Delta\tilde{W}|\leq w/200$.
The intuition of the proof is that if we could find such a $\tilde{W}$ then we can ``correct'' it to $W$ itself using a small number of Grover searches, so finding a good approximation~$\tilde{W}$ is not much easier than finding~$W$ itself.

\begin{theorem}\label{thm:lower_approximate_multi_search}
Let $w$ be an integer satisfying $0<w\leq d/2$, $W\subset \intg_d$ with size $w$, and $x\in\{0,1\}^d$ such that $x_j=1$ if $j\in W$ and $x_{j'}=0$ if $j'\in\overline{W}$. Suppose we have query access to $x$. Then every bounded-error quantum algorithm that outputs $\tilde{W}\subset \intg_d$ satisfying $|W\Delta \tilde{W}| \leq w/200$ makes $\Omega(\sqrt{dw})$ queries.
\end{theorem}

\begin{proof}
Suppose there exists a $T$-query bounded-error quantum algorithm to find a set $\tilde{W}$ satisfying $|W\Delta \tilde{W}| \leq w/200$. Define a function $f$ which marks the elements of $F=W\Delta \tilde{W}$. Since we have a classical description of $\tilde{W}$, we can implement a query to $f$ using one query to $x$.
Now use Corollary~\ref{thm:findallsolutions} (with $u=w/200$) to find all elements of $F$ with probability~1, using $\frac{\pi}{2}\sqrt{dw/200}+w/200$ queries. This gives a bounded-error quantum algorithm that finds $W$ itself using $T'=T + \frac{\pi}{2}\sqrt{dw/200}+w/200$ queries. By Theorem~\ref{thm:LB_multi_search_exact} we have $T' \geq \frac{1}{8}\sqrt{dw}$, implying $T= \Omega(\sqrt{dw})$.
\end{proof}

Next we consider the \emph{Hamming-weight distinguisher problem} $\mbox{HD}_{\ell,\ell'}$: given a $z\in\{0,1\}^N$ of Hamming weight $\ell$ or $\ell'$, distinguish these two cases. 
%The output of $\mbox{HD}_{\ell,\ell'}$ is either $0$ or $1$, where $1$ for the $\ell'$ case and $0$ for the other case.
The adversary bound gives the following bound (a special case of a result of Nayak and Wu~\cite{NW99} based on the polynomial method~\cite{bbcmw:polynomialsj}).

\begin{theorem}\label{thm:nayak_wu}
Let $N\in 2\intg_+$, $z\in \{0,1\}^N$, and $p\in (0,0.5)$ be multiple of $1/N$. Suppose we have query access to $z$.
Then every bounded-error quantum algorithm that computes $\mbox{HD}_{\frac{N}{2},N(\frac{1}{2}+p)}$ makes $\Omega(1/p)$ queries.
\end{theorem}

\begin{proof}
Let $\cX=\{x\in\{0,1\}^N: |x|=N/2\}$, $\cY=\{y\in \{0,1\}^N:|y|=N/2+pN\}$, and consider the relation $\cR=\{(x,y):x\in\cX,y\in\cY,x\leq y\}$, where $x\leq y$ if and only if $\forall i\in \intg_N$, $x_i\leq y_i$. We know
\begin{itemize}
    \item For every $x \in \cX$, there exist at least $m=\binom{N/2}{pN}$ different $y \in \cY$ such that $(x, y) \in R$.
    \item For every $y \in \cY$, there exist at least $m'=\binom{N/2+pN}{pN}$ different $x \in \cX$ such that $(x, y) \in R$.
    \item  $\ell_{max}=\max\limits_{\substack{(x, y) \in \cR, j \in\intg_d\\ s.t. x_j \neq y_j}}\ell_{x,j}\cdot \ell_{y,j}=\binom{N/2-1}{pN-1}\cdot \binom{N/2+pN-1}{pN-1}$.
\end{itemize}
Hence by Theorem~\ref{thm:adv_bound}, every bounded-error quantum algorithm that computes $g$ uses at least 
$$
\Big(\frac{1}{2}-\sqrt{\frac{9}{10}\cdot \frac{1}{10}}\Big)\cdot \sqrt{\frac{mm'}{\ell_{max}}}=\Omega\left(\sqrt{\frac{\binom{N/2}{pN}\cdot\binom{N/2+pN}{pN}}{\binom{N/2-1}{pN-1}\cdot \binom{N/2+pN-1}{pN-1}}}\right)=\Omega\left(\sqrt{\frac{N/2\cdot(N/2+pN)}{pN\cdot pN}}\right)=\Omega\Big(\frac{1}{p}\Big)
$$ 
queries.
\end{proof}

The above theorem implies a lower bound of $\Omega(1/p)$ queries for $\mbox{HD}_{\frac{N}{2},N(\frac{1}{2}+p)}$. One can also think of the input bits as $\pm 1$ and in this case, the goal is to distinguish whether the entries add up to $0$ or to $2pN$.
For convenience, we abuse the notation $\mbox{HD}_{\frac{N}{2},N(\frac{1}{2}+p)}$ also for the problem with $\pm 1$ inputs. Now we are ready to prove a lower bound for the \emph{worst-case set-finding problem} $\mbox{WSF}_{d,w,p,N}$: given a matrix $X\in\{-1,1\}^{N\times d}$ where each column-sum is either $2pN$ or $0$, the goal is to find a set $\tilde{W}\subset \intg_d$ such that $|\tilde{W}\Delta W|\leq w/200$, where $W$ is the set of indices for those columns whose entries add up to $2pN$ and $w=|W|$. One can see that this problem is actually a composition of the approximate set-finding problem and the Hamming-weight distinguisher problem. Composing the relational problem $\mbox{ASF}_{d,w}$ with $d$ valid inputs of $\mbox{HD}_{\frac{N}{2},N(\frac{1}{2}+p)}$, exactly $w$ of which evaluate to~1, we can see that the $d$-bit string given by the values of $\mbox{HD}_{\frac{N}{2},N(\frac{1}{2}+p)}$ on these $d$ inputs, is a valid input for $\mbox{ASF}_{d,w}$. In other words, the set of valid inputs for $\mbox{WSF}_{d,w,p,N}$, or equivalently, the set of valid inputs for the composed problem $\mbox{ASF}_{d,w}\circ (\mbox{HD}_{\frac{N}{2},N(\frac{1}{2}+p)})^d$ is
\[
\{(x^{(1)},\ldots,x^{(d)})\in\mathcal{P}^d:|\mbox{HD}_{\frac{N}{2},N(\frac{1}{2}+p)}(x^{(1)})\ldots \mbox{HD}_{\frac{N}{2},N(\frac{1}{2}+p)}(x^{(d)})|=w\},
\]
where $\mathcal{P}=\{x\in \{0,1\}^N:|x|\in\{N/2,N/2+pN\}\}$. 
The next theorem by Belovs and Lee shows that the quantum query complexity of the composed problem $\mbox{ASF}_{d,w}\circ (\mbox{HD}_{\frac{N}{2},\frac{N+2pN}{2}})^d$ is at least the product of the complexities of the two composing problems:

\begin{theorem}[\cite{BL20}, Corollary~27]\label{thm:composition}
Let $f\subseteq S\times T$, with $S \subseteq \{0,1\}^d$, be a relational problem with bounded-error quantum query complexity $L$. Assume that $f$ is efficiently verifiable, that is given some $t \in T$ and oracle access to $x \in S$, there exists a bounded-error quantum algorithm that verifies whether $(x, t) \in f$ using $o(L)$ queries to $x$. Let $D \subseteq \{0,1\}^N$ and $g : D \rightarrow \{0,1\}$ be a Boolean function whose bounded-error quantum query complexity is $Q$. Then the bounded-error quantum query complexity of the relational problem $f \circ g^d$, restricted to inputs $x \in \{0, 1\}^{dN}$ such that $g^d(x) \in S$, is $\Omega(LQ)$.
\end{theorem}

Applying Theorem~\ref{thm:composition} with the lower bounds of Theorem~\ref{thm:nayak_wu} and Theorem~\ref{thm:lower_approximate_multi_search}, we obtain:

\begin{corollary}\label{cor:WSF}
Let $N\in 2\intg_+$ and $p\in (0,0.5)$ be an integer multiple of $1/N$. Given a matrix $X \in \{-1,+1\}^{N\times d}$ such that there exists a set $W\subseteq \intg_d$ with size $w$ and \begin{itemize}
    \item For every $j \in W$, $\sum\limits_{i\in\intg_N} X_{ij}=2pN$.
    \item For every $j' \in \overline{W}$, $\sum\limits_{i\in\intg_N} X_{ij'}= 0$.
\end{itemize}
Suppose we have query access to $X$. Then every bounded-error quantum algorithm that computes $\tilde{W}$ such that $|W\Delta \tilde{W}|\leq w/200$, uses $\Omega(\sqrt{dw}/p)$ queries to $O_X$.
\end{corollary}

\subsection{Worst-case to average-case reduction for the set-finding problem}\label{sec:worst_to_average}
Our goal is to prove a lower bound for Lasso algorithms that have high success probability w.r.t.\ the distribution $\mathcal{D}_{p,W}$,
yet the lower bound of the previous subsection is for \emph{worst-case} instances.
In this subsection, we will connect these by providing a worst-case to average-case reduction for the set-finding problem. After that, by simply combining with the query lower bound for the worst-case set-finding problem and the reduction from the distributional set-finding problem to Lasso, we obtain an $\Omega({\sqrt{d}/\epsilon^{1.5}})$ query lower bound for Lasso.

\begin{theorem}\label{thm:worst_to_average}
Let $N\in 2\intg_+$, $p\in (0,0.5)$ be an integer multiple of $1/N$, $w$ be a natural number between $2$ to $d/2$, and $M$ be a natural number. Suppose $X\in \{-1,+1\}^{N\times d}$ is a valid input for $\mbox{WSF}_{d,w,p,N}$, and let $W\subset\intg_d$ be the set of the $w$ indices of the columns of $X$ whose entries add up to $2pN$. Let $R\in \intg_N^{M\times d}$ be a matrix whose entries are i.i.d.\ samples from $\cU_N$, and define $X'\in \{-1,1\}^{M\times d}$ as $X'_{ij}=X_{R_{ij}j}$. Then the $M$ vectors $(X'_i,1)$, where $X'_i$ is the $i$th row of $X'$ and $i\in \intg_M$, are i.i.d.\ samples from $\mathcal{D}_{p,W}$.
\end{theorem}

\begin{proof}
Every entry of $R$ is a sample from $\cU_N$, so $X_{R_{ij}j}$ is uniformly chosen from the entries of the $j$th column of $X$. Moreover, because every valid input $W$ for $\mbox{WSF}_{d,w,p,N}$ satisfies that for every $j\in 
W$, $\Pr_{i\sim \cU_N}[X_{ij}=1]=1/2+p$ and for every $j'\in \overline{W}$, $\Pr_{i\sim \cU_N}[X_{ij'}=1]=1/2$, we know $(X'_i,1)$ is distributed as $D_{p,W}$.
\end{proof}

The above theorem tells us that we can convert an instance of $\mbox{WSF}_{d,w,p,N}$ to an instance of $\mbox{DSF}_{\mathcal{D}_{p,W}}$. Note that we can produce matrix $R$ offline and therefore we can construct the oracle $O_{X'}:\ket{i}\ket{j}\ket{0}\rightarrow \ket{i}\ket{j}\ket{X_{R_{ij}j}}$ using $1$ query to $O_X:\ket{i}\ket{j}\ket{0}\rightarrow \ket{i}\ket{j}\ket{X_{ij}}$ (and some other elementary gates, which is irrelevant to the number of queries). Also observe that if $M=10^{12}\cdot\lceil \log d\rceil\cdot \lfloor 1/\epsilon\rfloor^2=\mathcal{O}((\log d)/\epsilon^2)$ and hence $S'=\{(X'_i,1)\}_{i=0}^{M-1}$ is a sample set with $M$ i.i.d.\ samples from $\mathcal{D}_{p,W}$, then by Theorem~\ref{thm:Rademacher_Lasso}, with probability $\geq 9/10$, an $\epsilon/16000$-minimizer for Lasso with respect to $S'$ is also an $\epsilon/8000$-minimizer for Lasso with respect to distribution $\mathcal{D}_{p,W}$. 
%Therefore we know that every algorithms that solves the distributional set-finding problem also solves the worst-case set-finding problem. 
By Theorem~\ref{thm:average_HSF_to_Lasso}, an $\epsilon/8000$-minimizer for Lasso with respect to distribution $\mathcal{D}_{p,W}$ can be used to output a set $\tilde{W}\subset \intg_d$ such that $|\tilde{W}\Delta W|\leq w/200$, where $W$ is the set of indices for those columns of $X$ whose entries add up to $2pN$. Hence we have a reduction from the worst-case set-finding problem to Lasso. By the reduction above and by plugging $w=\lfloor 1/\epsilon\rfloor$ and $p=1/(2\lfloor 1/\epsilon\rfloor)$ in Corollary~\ref{cor:WSF} (and $N$ an arbitrary natural number such that $pN \in \mathbb{N}$), we obtain a lower bound of $\Omega(\sqrt{d}/\epsilon^{1.5})$ queries for $\mbox{WSF}_{d,w,p,N}$, and hence the main result of this section: a lower bound of $\Omega(\sqrt{d}/\epsilon^{1.5})$ for Lasso. 

\begin{corollary}
Let $\epsilon\in (2/d, 1/100)$, $w=\lfloor 1/\epsilon\rfloor$, $p=1/(2\lfloor 1/\epsilon\rfloor)$, and $W\subset \intg_d$ with size $w$. Every bounded-error quantum algorithm that computes an $\epsilon$-minimizer for Lasso with respect to $\mathcal{D}_{p,W}$ uses $\Omega(\sqrt{d}/\epsilon^{1.5})$ queries.
\end{corollary}

\paragraph{Classical lower bound.}
In Appendix~\ref{app:classicalLBforLasso} we show how this quantum lower bound approach can be modified to prove, for the first time, a lower bound of $\tilde{\Omega}(d/\eps^2)$ on the classical query complexity of Lasso. This lower bound is optimal up to logarithmic factors.

\section{Quantum query lower bound for Ridge}\label{sec:distributed_HSF_to_Ridge}

Now we switch our attention from Lasso to Ridge. We will prove a lower bound of $\Omega(d/\epsilon)$ queries for Ridge in a very similar way as our lower bound for Lasso. Recall that Ridge's setup assumes the vectors in the sample set are normalized in $\ell_2$ rather than $\ell_\infty$ as in Lasso.
We modify the distribution to $\mathcal{D}'_{p,W}$ over $(x,y)\in\{-1/\sqrt{d},1/\sqrt{d}\}^d\times \{-1,1\}$ as follows. Let $p\in (0,1/4)$, $W\subset \mathbb{Z}_d$, and $\overline{W}=\mathbb{Z}_d\setminus W$. For each $j'\in \overline{W}$, $x_{j'}$ is generated according to  $\Pr[x_{j'}=-1/\sqrt{d}]=1/2+p$; for each $j\in W$, $x_j$ is generated according to  $\Pr[x_{j}=1/\sqrt{d}]=1/2+p$; $y$ is generated according to $\Pr[y=1]=1$. Now again we want to solve a distributional set-finding problem with respect to $\mathcal{D}'_{p,W}$, given $M$ samples from $\mathcal{D}'_{p,W}$. Similar to the Lasso case, one can think of the $M\times d$ matrix of samples as ``hiding'' the set $W$: the columns corresponding to $j\in W$ are likely to have more $1/\sqrt{d}$'s than $-1/\sqrt{d}$'s, while the columns corresponding to $j\in \overline{W}$ are likely to have more $-1/\sqrt{d}$'s than $1/\sqrt{d}$'s. 

In this section let 
\[
\theta^*=\sum\limits_{j\in\intg_d}\frac{e_j}{\sqrt{d}}(-1)^{[j\in\overline{W}]}
\]
and note that for every $\theta \in \mathbb{R}^d$,
\begin{align*}
L_{\mathcal{D}'_{p,W}}(\theta) &=  \E_{(x,y)\sim \mathcal{D}'_{p,W}}[\inProd{\theta}{x}^2]-2\E_{(x,y)\sim \mathcal{D}'_{p,W}}[\inProd{\theta}{x}]+1\\
&=  (\E_{(x,y)\sim \mathcal{D}'_{p,W}}[\inProd{\theta}{x}^2]-\E_{(x,y)\sim \mathcal{D}'_{p,W}}[\inProd{\theta}{x}]^2)+\E_{(x,y)\sim \mathcal{D}'_{p,W}}[\inProd{\theta}{x}]^2-2\E_{(x,y)\sim \mathcal{D}'_{p,W}}[\inProd{\theta}{x}]+1\\
&= \|\theta\|^2_2\cdot(1-4p^2)/d+(\E_{(x,y)\sim \mathcal{D}'_{p,W}}[\inProd{\theta}{x}]-1)^2\\
&=\|\theta\|^2_2\cdot(1-4p^2)/d +(2p\inProd{\theta}{\theta^*}-1)^2,
\end{align*}
where the third equality holds because $\inProd{\theta}{x}$ is a sum of independent random variables and hence its variance is the sum of the variances of the terms $\theta_ix_i$ (which are $\theta_i^2(1-4p^2)/d$). 

Next we show that $\theta^*$ is the minimizer for Ridge with respect to ${\mathcal{D}'_{p,W}}$.

\begin{theorem}
Let $w=\lfloor d/2\rfloor$ and $W \subset \mathbb{Z}_d$ be a set of size $w$, and let $\epsilon \in (1000/d,1/10000)$ and $p=1/\lfloor 1/\epsilon\rfloor$. Then $\theta^*=\sum\limits_{j\in\intg_d}\frac{e_j}{\sqrt{d}}(-1)^{[j\in\overline{W}]}$ is the minimizer for Ridge with respect to ${\mathcal{D}'_{p,W}}$.
\end{theorem}

\begin{proof}
Let $\theta=\sum\limits_{j\in\intg_d}\theta_j e_j\in B^d_2$ be a minimizer. We want to show $\theta_j=\theta^*_j$ for every $j\in \intg_d$. Note that if $\theta_j\cdot (-1)^{[j\in \overline{W}]}<0$, then we can flip the sign of $\theta_j$ to get a smaller objective value, that is,
\begin{align*}
    L_{\mathcal{D}'_{p,W}}(\theta')-L_{\mathcal{D}'_{p,W}}(\theta)&=(\|\theta'\|^2_2-\|\theta\|^2_2)\cdot(1-4p^2)/d +(2p\inProd{\theta'}{\theta^*}-1)^2-(2p\inProd{\theta}{\theta^*}-1)^2 \\
    &=(2p\inProd{\theta'-\theta}{\theta^*})(2p\inProd{\theta'+\theta}{\theta^*}-2)\\
    &=(-4p\theta_j\cdot (-1)^{[j\in\overline{W}]})(2p\inProd{\theta'+\theta}{\theta^*}-2) <0,
\end{align*}
where $\theta'=\sum\limits_{k\in\intg_d\setminus\{j\}}\theta_ke_k-\theta_je_j$, and the last inequality is because $-4p\theta_j\cdot (-1)^{[j\in\overline{W}]}>0$ and $2p\inProd{\theta'+\theta}{\theta^*}\leq2p\|\theta'+\theta\|_2\cdot\|\theta^*\|_2 \leq 4p \leq 1$. Since $\theta$ was assumed a minimizer,  for all $j\in \intg_d$ the sign of $\theta_j$ must be $(-1)^{[j\in \overline{W}]}$.

Second, we show that we must have $|\theta_0|=|\theta_1|=\cdots=|\theta_{d-1}|$. 
Suppose, towards a contradiction, that this is not the case. Consider $\theta'=\sum\limits_{j\in\intg_d} ue_j\cdot (-1)^{[j\in\overline{W}]}$, where $u=\sqrt{ \sum\limits_{j\in\intg_d} |\theta_j|^2/d}$. We have
\begin{align*}
    L_{\mathcal{D}'_{p,W}}(\theta')-L_{\mathcal{D}'_{p,W}}(\theta)    &=(2p\inProd{\theta'-\theta}{\theta^*})(2p\inProd{\theta'+\theta}{\theta^*}-2)\\
    &=(2p/\sqrt{d})\cdot (du-\sum\limits_{j\in\intg_d}|\theta_j|)\cdot(2p\inProd{\theta'+\theta}{\theta^*}-2) <0.
\end{align*}
The last inequality holds because again $2p\inProd{\theta'+\theta}{\theta^*}\leq 4p \leq1$ and in addition, $$d\cdot \sum\limits_{j\in\intg_d} |\theta_j|^2 >(\sum\limits_{j\in\intg_d}|\theta_j|)^2$$ by the Cauchy–Schwarz inequality (which is strict if the $|\theta_j|$ are not all equal). Hence if $\theta$ is indeed a minimizer, then its entries must all have the same magnitude.

Now we know a minimizer $\theta$ must be in the same direction as $\theta^*$, we just don't know yet that the magnitudes of its entries are $1/\sqrt{d}$. Suppose $\|\theta\|_2=u\leq1$ and $\theta=u\cdot\theta^*$, then we have
\begin{align*}
      L_{\mathcal{D}'_{p,W}}(\theta)  &=  \|\theta\|^2_2\cdot(1-4p^2)/d +(2p\inProd{\theta}{\theta^*}-1)^2 \\
    &=(u^2(1-4p^2)/d+(2pu-1)^2).
\end{align*}
The discriminant of $f(u)=u^2(1-4p^2)/d+(2pu-1)^2$ is less than $0$, and $u=\frac{2p}{4p^2+(1-4p^2)/d}$ is the global minimizer of $f(u)$. Note that $u=\frac{2p}{4p^2+(1-4p^2)/d}>1$, and hence $f(1)\leq f(u)$ for every $u\leq 1$. Therefore we know $\theta^*$ is the minimizer for Ridge with respect to $\mathcal{D}'_{p,W}$.
\end{proof}

Next we show that the inner product between the minimizer and an approximate minimizer for Ridge will be close to $1$.

\begin{theorem}\label{thm:weight_distinguisher_Ridge}
Let $w=\lfloor d/2\rfloor$, $W \subset \mathbb{Z}_d$ be a set of size $w$, $\epsilon \in (1000/d,1/10000)$, and $p=1/\lfloor 1/\epsilon\rfloor$. Suppose $\theta\in B^d_2$ is an $\epsilon/1000$-minimizer for Ridge with respect to $\mathcal{D}'_{p,W}$. Then $\inProd{\theta}{\theta^*}\geq 0.999$.
\end{theorem}

\begin{proof}
Because $\theta$ is an $\epsilon/1000$-minimizer, we have 
\begin{align*}
    & 0.001\epsilon\geq L_{\mathcal{D}'_{p,W}}(\theta)-L_{\mathcal{D}'_{p,W}}(\theta^*) = (1-4p^2)\cdot(\|\theta\|^2_2-1)/d+(2p\inProd{\theta}{\theta^*}-1)^2-(2p-1)^2\\
    \implies & 2p\inProd{\theta}{\theta^*}\geq 1-\sqrt{1-4p+4p^2+0.001\epsilon-(1-4p^2)\cdot(\|\theta\|^2_2-1)/d}.
\end{align*}
Letting $z=4p-4p^2-0.001\epsilon+(1-4p^2)\cdot(\|\theta\|^2_2-1)/d$, we have
\begin{align*}
2p\inProd{\theta}{\theta^*} \geq & 1-\sqrt{1-z}\geq  1-(1-z/2) =  z/2\\
    =& 2p-2p^2+(1-4p^2)\cdot(\|\theta\|^2_2-1)/d-0.001\epsilon,
\end{align*}
where the second inequality holds because $z\in(0,1)$.
Dividing both sides by $2p$, we have
\begin{align*}
    \inProd{\theta}{\theta^*}\geq 1-p+(1-4p^2)\cdot(\|\theta\|^2_2-1)/(2pd)-0.0005\epsilon/p.
\end{align*}
Because $\theta\in B^d_2$, $p=1/\lfloor 1/\epsilon\rfloor$, and $\epsilon \in (1000/d,1/10000)$, the above implies $\inProd{\theta}{\theta^*}\geq0.999$.
\end{proof}

Combining the above theorem with the following theorem, we can see how to relate the entries of an approximate minimizer for Ridge with respect to $\mathcal{D}'_{p,W}$ to the elements of the hidden set $W$.

\begin{theorem}\label{thm:Ridge_inner_product}
Suppose $\theta \in B^d_2$ satisfies $\inProd{\theta}{\theta^*}\geq 1-0.001$. Then $\#\{j\in\intg_d \mid \theta_j\cdot \theta^*_j \leq 0\} \leq d/500$.
\end{theorem}

\begin{proof}
If $\theta_j\cdot \theta_j^*\leq 0$ then $|\theta_j-\theta_j^*|\geq |\theta_j^*|=\frac{1}{\sqrt{d}}$, hence using Theorem~\ref{thm:weight_distinguisher_Ridge} we have
\begin{align*}
    \frac{1}{d} \#\{j\in\intg_d \mid \theta_j\cdot \theta^*_j \leq 0\}\leq \|\theta-\theta^*\|_2^2 = \|\theta\|^2_2+\|\theta^*\|^2_2-2\inProd{\theta}{\theta^*}\leq 2-2(1-0.001)=1/500.
\end{align*}
\end{proof}

We know $\theta^*=\sum\limits_{j\in\intg_d}\frac{e_j}{\sqrt{d}}(-1)^{[j\in\overline{W}]}$, so by looking at the signs of entries of $\theta$, we can find an index set $\tilde{W}=\{j\in\intg_d:\theta_j >0\}$ satisfying that $|W\Delta \tilde{W}|\leq d/500\leq w/200$ because $w=\lfloor d/2\rfloor$. Therefore, once we have an $\epsilon/1000$-minimizer for Ridge with respect to $\mathcal{D}'_{p,W}$, we can solve $\mbox{DSF}_{\mathcal{D}'_{p,W}}$.

With the reduction from $\mbox{DSF}_{\mathcal{D}'_{p,W}}$ to Ridge, we here show (similar to Lasso) a lower bound for the \emph{worst-case symmetric set-finding problem} $\mbox{WSSF}_{d,w,p,N}$: given a matrix $X\in\{-1/\sqrt{d},1/\sqrt{d}\}^{N\times d}$ where each column-sum is either $2pN/\sqrt{d}$ or $-2pN/\sqrt{d}$, the goal is to find a set $\tilde{W}\subset \intg_d$ such that $|\tilde{W}\Delta W|\leq w/200$, where $W$ is the set of indices for those columns whose entries add up to $2pN/\sqrt{d}$ and $w=|W|$. This problem is again a composition of the approximate set finding problem in Section~\ref{sec:q_lower_decomposed} and the Hamming-weight distinguisher problem $\mbox{HD}_{\ell,\ell'}$ with $\ell=\frac{N}{2}-pN$ and $\ell'=\frac{N}{2}+pN$ up to a scalar $1/\sqrt{d}$. Following the proof of Theorem~\ref{thm:nayak_wu}, we prove a lower bound of $\Omega(1/p)$ queries for this problem.

\begin{theorem}
Let $N\in 2\intg_+$, $z\in \{0,1\}^N$, and $p\in (0,0.5)$ be an integer multiple of $1/N$. Suppose we have query access to $z$. 
Then every bounded-error quantum algorithm that computes $\mbox{HD}_{\frac{N}{2}-pN,\frac{N}{2}+pN}$ makes $\Omega(1/p)$ queries.
\end{theorem}

Again we think of the input bits as $\pm 1$ and abuse the notation $\mbox{HD}_{\frac{N}{2}-pN,\frac{N}{2}+pN}$ for the problem with $\pm 1$ input. Also, by the composition property of the adversary bound from Belovs and Lee~\cite{BL20} (Theorem~\ref{thm:composition}), we have a lower bound of $\Omega(\sqrt{dw}/p)$ for $\mbox{WSSF}_{d,w,p,N}$ from the $\Omega(\sqrt{dw})$ lower bound for $\mbox{ASF}_{d,w}$ and the $\Omega(1/p)$ lower bound for $\mbox{HD}_{\frac{N}{2}-pN,\frac{N}{2}+pN}$.

\begin{corollary}\label{cor:WSF_Ridge}
Let $N\in 2\intg_+$ and $p\in (0,0.5)$ be an integer multiple of $1/N$. Given a matrix $X \in \{-1/\sqrt{d},+1/\sqrt{d}\}^{N\times d}$ such that there exists a set $W\subseteq \intg_d$ with size $w$ and \begin{itemize}
    \item For every $j \in W$, $\sum\limits_{i\in\intg_N} X_{ij}=2p  N/\sqrt{d}$.
    \item For every $j' \in \overline{W}$, $\sum\limits_{i\in\intg_N} X_{ij'}= -2p  N/\sqrt{d}$.
\end{itemize}
Then every bounded-error quantum algorithm that computes $\tilde{W}$ such that $|W\Delta \tilde{W}|\leq w/200$, takes $\Omega(\sqrt{dw}/p)$ queries.
\end{corollary}

The final step for proving a lower bound for Ridge, using the same arguments as in Section~\ref{sec:worst_to_average}, is to provide a worst-case to average-case reduction for the symmetric set-finding problem. We follow the same proof in Theorem~\ref{thm:worst_to_average} and immediately get the following theorem: %This follows the same proof in Theorem~\ref{thm:worst_to_average}: we first generate a random matrix $R$ such that every entry of $R$ is a sample from $\cU_N$, so $X_{R_{ij,j}}$ is uniformly chosen from the entries of $j$th column of $X$, where $X$ is a valid input for $\mbox{WSSF}_{d,w,p,N}$. Moreover, because every valid input $X$ for $\mbox{WSSF}_{d,w,p,N}$ satisfies $\Pr_{i\sim \cU_N}[X_{ij}=1/\sqrt{d}]=1/2+p$ for every $j\in 
%W$, and $\Pr_{i\sim \cU_N}[X_{ij'}=1/\sqrt{d}]=1/2-p$ for every $j'\in \overline{W}$, we know $(X'_i,1)$ is distributed as $\mathcal{D}'_{p,W}$. This implies the following theorem:

\begin{theorem}
Let $N\in 2\intg_+$, $p\in (0,0.5)$ be an integer multiple of $1/N$, $w$ be a natural number between $2$ to $d/2$, and $M$ be a natural number. Suppose $X\in \{-1/\sqrt{d},+1/\sqrt{d}\}^{N\times d}$ is a valid input for $WSSF_{d,w,p,N}$, and let $W\subset\intg_d$ be the set of the $w$ indices of the columns of $X$ whose entries add up to $2pN/\sqrt{d}$. Let $R\in \intg_N^{M\times d}$ be a matrix whose entries are i.i.d.\ samples from $\cU_N$, and define $X'\in \{-1/\sqrt{d},1/\sqrt{d}\}^{M\times d}$ as $X'_{ij}=X_{R_{ij}j}$. Then the vectors $(X'_i,1)$, where $X'_i$ is the $i$th row of $X'$ and $i\in \intg_M$, are i.i.d.\ samples from $\mathcal{D}'_{p,W}$.
\end{theorem}

By setting $M=10^{10}\cdot\lceil \log d\rceil\cdot \lfloor 1/\epsilon\rfloor^2=\mathcal{O}((\log d)/\epsilon^2)$ and letting  $S'=\{(X'_i,1)\}_{i=0}^{M-1}$ be a sample set with $M$ i.i.d.\ samples from $\mathcal{D}'_{p,W}$, with probability $\geq 9/10$, an $\epsilon/2000$-minimizer for Ridge with respect to $S'$ is also an $\epsilon/1000$-minimizer for Ridge with respect to distribution $\mathcal{D}'_{p,W}$ from Theorem~\ref{thm:Rademacher_Ridge}. 
%Therefore we know that every algorithms that solves the distributional set-finding problem also solves the worst-case set-finding problem. 
By Theorem~\ref{thm:Ridge_inner_product} and Theorem~\ref{thm:weight_distinguisher_Ridge}, an $\epsilon/1000$-minimizer for Ridge with respect to distribution $\mathcal{D}'_{p,W}$ gives us a set $\tilde{W}\subset \intg_d$ such that $|\tilde{W}\Delta W|\leq w/200$, where $W$ is the set of indices for those columns of $X$ whose entries add up to $2pN/\sqrt{d}$. Hence we have a reduction from the worst-case symmetric set-finding problem to Ridge. By this reduction and by plugging $w=\lfloor d/2\rfloor$ and $p=1/\lfloor 1/\epsilon\rfloor$ in Corollary~\ref{cor:WSF_Ridge} (and $N$ an arbitrary natural number such that $pN \in \mathbb{N}$), we obtain a lower bound of $\Omega(d/\epsilon)$ queries for $\mbox{WSSF}_{d,w,p,N}$, and hence for Ridge as well, which is the main result of this section.

\begin{corollary}
Let $\epsilon\in (2/d, 1/1000)$, $w=\lfloor d/2\rfloor$, $p=1/\lfloor 1/\epsilon\rfloor$, and $W\subset \intg_d$ with size $w$. Every bounded-error quantum algorithm that computes an $\epsilon$-minimizer for Ridge with respect to $\mathcal{D}'_{p,W}$ uses $\Omega(d/\epsilon)$ queries.
\end{corollary}

\section{Future work}

We mention a few directions for future work:
\begin{itemize}
\item While the $d$-dependence of our quantum bounds for Lasso is essentially optimal, the $\eps$-dependence is not: upper bound $\sqrt{d}/\eps^2$ vs lower bound $\sqrt{d}/\eps^{1.5}$.
Can we shave off a $1/\sqrt{\eps}$ factor from our upper bound, maybe using a version of accelerated gradient descent~\cite{Nesterov1983AMF} with $O(1/\sqrt{\eps})$ iterations instead of Frank-Wolfe's $O(1/\eps)$ iterations? 
Or can we somehow improve our lower bound by embedding harder query problems into Lasso?
Recall from Table~\ref{tab:linear_regression} that even for classical algorithms the optimal $\eps$-dependence seems to be open; it might be possible to get a tight classical lower bound by a classical version of our quantum lower bound, but it remains to be worked out whether the required composition property (the classical analogue of Theorem~\ref{thm:composition}) holds.
\item Similar question for Ridge: the linear $d$-dependence of our quantum bounds is tight, but we should improve the $\eps$-dependence of our upper and/or lower bounds. The most interesting outcome would be a quantum algorithm for Ridge with better $\eps$-dependence than the optimal classical complexity of $\tilde{\Theta}(d/\eps^2)$; currently we do not know of any quantum speedup for Ridge.
    \item Can we speed up some other methods for (smooth) convex optimization? In particular, can we find a classical iterative method where quantum algorithms can significantly reduce the number of iterations, rather than just the cost per iteration as we did here?
    \item There are many connections between Lasso and Support Vector Machines~\cite{jaggi:lassosvm}, and there are recent quantum algorithms for optimizing SVMs~\cite{rebentrost2014QSVM,schuld&killoran:feature,saeedi&arodz:qsvm,allcock&hsieh:qsvm,SPA:semi}. We would like to understand this connection better.
\end{itemize}

\subsubsection*{Acknowledgements.}
We thank Yi-Shan Wu and Christian Majenz for useful discussions, and Armando Bellante for pointing us to~\cite{BZ:qmpursuit}.

\bibliographystyle{alpha}
\bibliography{ref}

\newcommand{\etalchar}[1]{$^{#1}$}
\begin{thebibliography}{AGGW20b}

\bibitem[AG19a]{AG19zerosumgame}
{Joran van} Apeldoorn and Andr{\'a}s Gily{\'e}n.
\newblock Quantum algorithms for zero-sum games.
\newblock arXiv:1904.03180, 2019.

\bibitem[AG19b]{vAG19}
{Joran van} Apeldoorn and András Gilyén.
\newblock Improvements in quantum {SDP}-solving with applications.
\newblock In {\em Proceedings of 46th International Colloquium on Automata,
  Languages, and Programming}, volume 132 of {\em Leibniz International
  Proceedings in Informatics}, pages 99:1--99:15, 2019.
\newblock arXiv:1804.05058.

\bibitem[AGGW20a]{apeldoorn2018ConvexOptUsingQuantumOracles}
{Joran van} Apeldoorn, Andr{\'a}s Gily{\'e}n, Sander Gribling, and {Ronald de}
  Wolf.
\newblock Convex optimization using quantum oracles.
\newblock {\em Quantum}, 4:220, 2020.
\newblock arxiv:1809.00643.

\bibitem[AGGW20b]{vAGGdW17}
{Joran van} Apeldoorn, Andr{\'{a}}s Gily{\'{e}}n, Sander Gribling, and {Ronald
  de} Wolf.
\newblock Quantum {SDP}-solvers: better upper and lower bounds.
\newblock {\em Quantum}, 4:230, 2020.
\newblock Earlier version in FOCS'17. arXiv:1705.01843.

\bibitem[AGL{\etalchar{+}}21]{AGLNWW:scaling}
{Joran van} Apeldoorn, Sander Gribling, Yinan Li, Harold Nieuwboer, Michael
  Walter, and {Ronald de} Wolf.
\newblock Quantum algorithms for matrix scaling and matrix balancing.
\newblock In {\em Proceedings of 48th International Colloquium on Automata,
  Languages, and Programming}, volume 198 of {\em Leibniz International
  Proceedings in Informatics}, pages 110:1--17, 2021.
\newblock arXiv:2011.12823.

\bibitem[AH20]{allcock&hsieh:qsvm}
Jonathan Allcock and Chang-Yu Hsieh.
\newblock A quantum extension of {SVM-perf} for training nonlinear {SVM}s in
  almost linear time.
\newblock {\em Quantum}, 4:342, 2020.
\newblock arXiv:2006.10299.

\bibitem[AM20]{arunachalam&maity:qboosting}
Srinivasan Arunachalam and Reevu Maity.
\newblock Quantum boosting.
\newblock In {\em Proceedings of 37th International Conference on Machine
  Learning (ICML'20)}, 2020.
\newblock arXiv:2002.05056.

\bibitem[Amb02]{Amb00}
Andris Ambainis.
\newblock Quantum lower bounds by quantum arguments.
\newblock {\em Journal of Computer and System Sciences}, 64(4):750--767, 2002.
\newblock Earlier version in STOC'00. arXiv:quant-ph/0002066.

\bibitem[Ape20]{vApel20}
{Joran van} Apeldoorn.
\newblock {\em A quantum view on convex optimization}.
\newblock PhD thesis, Universiteit van Amsterdam, 2020.

\bibitem[AW20]{AW20}
Simon Apers and {Ronald de} Wolf.
\newblock Quantum speedup for graph sparsification, cut approximation and
  {L}aplacian solving.
\newblock In {\em Proceedings of 61st {IEEE} Annual Symposium on Foundations of
  Computer Science}, pages 637--648, 2020.
\newblock arXiv:1911.07306.

\bibitem[BBC{\etalchar{+}}01]{bbcmw:polynomialsj}
Robert Beals, Harry Buhrman, Richard Cleve, Michele Mosca, and {Ronald de}
  Wolf.
\newblock Quantum lower bounds by polynomials.
\newblock {\em Journal of the ACM}, 48(4):778--797, 2001.
\newblock Earlier version in FOCS'98. quant-ph/9802049.

\bibitem[BCWZ99]{bcwz:qerror}
Harry Buhrman, Richard Cleve, {Ronald de} Wolf, and Christof Zalka.
\newblock Bounds for small-error and zero-error quantum algorithms.
\newblock In {\em Proceedings of 40th IEEE FOCS}, pages 358--368, 1999.
\newblock cs.CC/9904019.

\bibitem[BG11]{BG11}
Peter B{\"u}hlmann and {Sara van de} Geer.
\newblock {\em Statistics for High-Dimensional Data: Methods, Theory and
  Applications}.
\newblock Springer, 2011.

\bibitem[BHMT02]{BHMT00}
Gilles Brassard, Peter H{\o}yer, Michele Mosca, and Alain Tapp.
\newblock Quantum amplitude amplification and estimation.
\newblock {\em Quantum Computation and Information}, page 53–74, 2002.
\newblock arXiv:quant-ph/0005055.

\bibitem[BHT98]{BrassardHT98}
Gilles Brassard, Peter H{\o}yer, and Alain Tapp.
\newblock Quantum counting.
\newblock In {\em Proceedings of 25th International Colloquium on Automata,
  Languages and Programming}, volume 1443 of {\em Lecture Notes in Computer
  Science}, pages 820--831, 1998.
\newblock arXiv:quant-ph/9805082.

\bibitem[BKL{\etalchar{+}}19]{brandao2017QSDPSpeedupsLearning}
Fernando Brand{\~a}o, Amir Kalev, Tongyang Li, Cedric Yen-Yu Lin, Krysta Svore,
  and Xiaodi Wu.
\newblock Quantum {SDP} solvers: Large speed-ups, optimality, and applications
  to quantum learning.
\newblock In {\em Proceedings of 46th International Colloquium on Automata,
  Languages, and Programming}, volume 132 of {\em Leibniz International
  Proceedings in Informatics}, pages 27:1--27:14, 2019.
\newblock arXiv:1710.02581.

\bibitem[BL20]{BL20}
Aleksandrs Belovs and Troy Lee.
\newblock The quantum query complexity of composition with a relation.
\newblock arXiv:2004.06439, 2020.

\bibitem[BS17]{brandao2016QSDPSpeedup}
Fernando Brand{\~a}o and Krysta Svore.
\newblock Quantum speed-ups for solving semidefinite programs.
\newblock In {\em Proceedings of 58th {IEEE} Annual Symposium on Foundations of
  Computer Science, {FOCS}}, pages 415--426, 2017.
\newblock arXiv:1609.05537.

\bibitem[BZ22]{BZ:qmpursuit}
Armando Bellante and Stefano Zanero.
\newblock Quantum matching pursuit: A quantum algorithm for sparse
  representations.
\newblock {\em Physical Review A}, 105:022414, 2022.

\bibitem[CCLW20]{chakrabarti2018QuantumConvexOpt}
Shouvanik Chakrabarti, Andrew Childs, Tongyang Li, and Xiaodi Wu.
\newblock Quantum algorithms and lower bounds for convex optimization.
\newblock {\em Quantum}, 4:221, 2020.
\newblock arXiv:1809.01731.

\bibitem[CGJ19]{CGJ18}
Shantanav Chakraborty, Andr{\'{a}}s Gily{\'{e}}n, and Stacey Jeffery.
\newblock The power of block-encoded matrix powers: improved regression
  techniques via faster {H}amiltonian simulation.
\newblock In {\em Proceedings of 46th International Colloquium on Automata,
  Languages, and Programming}, volume 132 of {\em Leibniz International
  Proceedings in Informatics}, pages 33:1--33:14, 2019.
\newblock arXiv:1804.01973.

\bibitem[CMP22]{CMP:qlinereg}
Shantanav Chakraborty, Aditya Morolia, and Anurudh Peduri.
\newblock Quantum regularized least squares.
\newblock arXiv:2206.13143, 2022.

\bibitem[CSS11]{CSS11}
Nicol{\`{o}} Cesa{-}Bianchi, Shai Shalev{-}Shwartz, and Ohad Shamir.
\newblock Efficient learning with partially observed attributes.
\newblock {\em Journal of Machine Learning Research}, 12:2857--2878, 2011.
\newblock arXiv:1004.4421.

\bibitem[DH96]{DH96}
Christoph D{\"u}rr and Peter H{\o}yer.
\newblock A quantum algorithm for finding the minimum.
\newblock quant-ph/9607014, 1996.

\bibitem[DHL{\etalchar{+}}20]{DHLYT20}
Yuxuan Du, Min{-}Hsiu Hsieh, Tongliang Liu, Shan You, and Dacheng Tao.
\newblock Quantum differentially private sparse regression learning.
\newblock arXiv:2007.11921, 2020.

\bibitem[FW56]{FW56}
Marguerite Frank and Philip Wolfe.
\newblock An algorithm for quadratic programming.
\newblock {\em Naval Research Logistics Quarterly}, 3(1‐2):95--110, 1956.

\bibitem[GKNS21a]{GKNS21porder}
Ankit Garg, Robin Kothari, Praneeth Netrapalli, and Suhail Sherif.
\newblock Near-optimal lower bounds for convex optimization for all orders of
  smoothness.
\newblock In {\em Proceedings of 35th Conference on Neural Information
  Processing Systems}, 2021.

\bibitem[GKNS21b]{GKNS21noquantumspeedup}
Ankit Garg, Robin Kothari, Praneeth Netrapalli, and Suhail Sherif.
\newblock No quantum speedup over gradient descent for non-smooth convex
  optimization.
\newblock In {\em Proceedings of 12th Innovations in Theoretical Computer
  Science Conference}, volume 185 of {\em Leibniz International Proceedings in
  Informatics}, pages 53:1--53:20, 2021.
\newblock arXiv:2010.01801.

\bibitem[GLT18]{GLT18}
Andr{\'{a}}s Gily{\'{e}}n, Seth Lloyd, and Ewin Tang.
\newblock Quantum-inspired low-rank stochastic regression with logarithmic
  dependence on the dimension.
\newblock arXiv:1811.04909, 2018.

\bibitem[GN22]{gribling&nieuwboer:scaling}
Sander Gribling and Harold Nieuwboer.
\newblock Improved quantum lower and upper bounds for matrix scaling.
\newblock In {\em Proceedings of 39th International Symposium on Theoretical
  Aspects of Computer Science (STACS 2022)}, volume 219 of {\em Leibniz
  International Proceedings in Informatics}, pages 35:1--35:23, 2022.
\newblock arXiv:2109.15282.

\bibitem[Gro96]{Grover96}
Lov Grover.
\newblock A fast quantum mechanical algorithm for database search.
\newblock In {\em Proceedings of 28th Annual {ACM} Symposium on the Theory of
  Computing}, pages 212--219, 1996.
\newblock arXiv:quant-ph/9605043.

\bibitem[HHL09]{HHL09}
Aram Harrow, Avinatan Hassidim, and Seth Lloyd.
\newblock Quantum algorithm for solving linear systems of equations.
\newblock {\em Physical Review Letters}, 103(15):150502, 2009.
\newblock arXiv:0811.3171.

\bibitem[HK70]{HK70}
Arthur Hoerl and Robert Kennard.
\newblock Ridge regression: biased estimation for nonorthogonal problems.
\newblock {\em Technometrics}, 12(1):55--67, 1970.

\bibitem[HK12]{HK11}
Elad Hazan and Tomer Koren.
\newblock Linear regression with limited observation.
\newblock In {\em Proceedings of the 29th International Conference on Machine
  Learning}, 2012.
\newblock arXiv:1206.4678. More extensive version at arXiv:1108.4559.

\bibitem[Hol03]{Holmes03}
Susan Holmes.
\newblock Stein’s method for birth and death chains.
\newblock In Persi Diaconis and Susan Holmes, editors, {\em Stein’s method:
  expository lectures and applications}. 2003.

\bibitem[IW20]{izdebski&wolf:qboosting}
Adam Izdebski and {Ronald de} Wolf.
\newblock Improved quantum boosting.
\newblock arXiv:2009.08360, 2020.

\bibitem[Jag13]{Jaggi13}
Martin Jaggi.
\newblock Revisiting {F}rank-{W}olfe: Projection-free sparse convex
  optimization.
\newblock In {\em Proceedings of the 30th International Conference on Machine
  Learning}, volume~28, pages 427--435, 2013.

\bibitem[Jag14]{jaggi:lassosvm}
Martin Jaggi.
\newblock An equivalence between the {Lasso} and {S}upport {V}ector {M}achines.
\newblock In Johan Suykens, Marco Signoretto, and Andreas Argyriou, editors,
  {\em Regularization, Optimization, Kernels, and Support Vector Machines}.
  2014.
\newblock arXiv:1303.1152.

\bibitem[KP17]{KP16}
Iordanis Kerenidis and Anupam Prakash.
\newblock Quantum recommendation systems.
\newblock In {\em Proceedings of 8th Innovations in Theoretical Computer
  Science Conference}, volume~67 of {\em Leibniz International Proceedings in
  Informatics}, pages 49:1--49:21, 2017.
\newblock arXiv:1603.08675.

\bibitem[MRT18]{MRT18}
Mehryar Mohri, Afshin Rostamizadeh, and Ameet Talwalkar.
\newblock {\em Foundations of Machine Learning}.
\newblock Adaptive Computation and Machine Learning series. MIT Press, second
  edition, 2018.

\bibitem[Nes83]{Nesterov1983AMF}
Yurii Nesterov.
\newblock A method for solving the convex programming problem with convergence
  rate $\mathcal{O}(1/k^2)$.
\newblock {\em Proceedings of the USSR Academy of Sciences}, 269:543--547,
  1983.

\bibitem[NW99]{NW99}
Ashwin Nayak and Felix Wu.
\newblock The quantum query complexity of approximating the median and related
  statistics.
\newblock In {\em Proceedings of the 31st Annual {ACM} Symposium on Theory of
  Computing}, pages 384--393. {ACM}, 1999.
\newblock arXiv:quant-ph/9804066.

\bibitem[Pra14]{Pra14}
Anupam Prakash.
\newblock {\em Quantum Algorithms for Linear Algebra and Machine Learning}.
\newblock PhD thesis, University of California, Berkeley, 2014.

\bibitem[RML14]{rebentrost2014QSVM}
Patrick Rebentrost, Masoud Mohseni, and Seth Lloyd.
\newblock Quantum support vector machine for big data classification.
\newblock {\em Physical Review Letters}, 113(13):130503, 2014.
\newblock arXiv:1307.0471.

\bibitem[SA19]{saeedi&arodz:qsvm}
Seyran Saeedi and Tom Arodz.
\newblock Quantum sparse support vector machines, 2019.
\newblock arXiv:1902.01879.

\bibitem[SB14]{SB14}
Shai Shalev{-}Shwartz and Shai Ben{-}David.
\newblock {\em Understanding Machine Learning - {F}rom Theory to Algorithms}.
\newblock Cambridge University Press, 2014.

\bibitem[SK19]{schuld&killoran:feature}
Maria Schuld and Nathan Killoran.
\newblock Quantum machine learning in feature {H}ilbert spaces.
\newblock {\em Physical Review Letters}, 122(13):040504, 2019.
\newblock arXiv:1803.07128.

\bibitem[SPA21]{SPA:semi}
Seyran Saeedi, Aliakbar Panahi, and Tom Arodz.
\newblock Quantum semi-supervised kernel learning.
\newblock {\em Quantum Machine Intelligence}, 3:24, 2021.

\bibitem[Tib96]{Tib96}
Robert Tibshirani.
\newblock Regression shrinkage and selection via the {L}asso.
\newblock {\em Journal of the Royal Statistical Society}, 58:267--288, 1996.

\bibitem[Vin78]{Vin78}
Hrishikesh Vinod.
\newblock A survey of {R}idge regression and related techniques for
  improvements over ordinary least squares.
\newblock {\em The Review of Economics and Statistics}, 60(1):121--131, 1978.

\bibitem[ZLL21]{ZLL:saddlepoints}
Chenyi Zhang, Jiaqi Leng, and Tongyang Li.
\newblock Quantum algorithms for escaping from saddle points.
\newblock {\em Quantum}, 5:229, 2021.
\newblock arXiv:2007.10253.

\end{thebibliography}

\appendix
\section{Proof of Theorem~\ref{thm:min_finding_approx}}\label{app:proof_approx_min}

Our proof of  Theorem~\ref{thm:min_finding_approx}
relies on the following result.

\begin{theorem}[\cite{vAGGdW17}, Theorem~49]\label{thm:generalized_min_finding} 
Let $m\in\mathbb{R}$ and $\delta_1 \in (0,1)$. Suppose we have a unitary $U$ that maps $\ket{0}\rightarrow \sum\limits_{\ell\in \intg_M}\sqrt{p_\ell}\ket{\psi_\ell}\ket{x_\ell}$, where the $\ket{\psi_\ell}$ are normalized states and the $x_\ell$ are real numbers satisfying $x_0<x_1<\cdots<x_{M-1}$, and define $X$ a random variable with $\Pr[X=x_\ell]=p_\ell$. Let $K$ be a natural number $\geq\frac{1000}{\sqrt{\Pr[X\leq m]}}\cdot \log(1/\delta_1)$. Then there exists a quantum algorithm that obtains a state $\ket{\psi_i}\ket{x_i}$ where $x_i\leq m$ with probability $\geq 
1-\delta_1$, using $K$ applications of $U$ and $U^\dagger$, and $\mathcal{\tilde{O}}(K)$ elementary gates.
\end{theorem}

%Now we want to do (approximate) minimum finding with an approximate oracle: let $v_0,v_1,\ldots,v_{d-1}\in\mathbb{R}$. Suppose we have an oracle $A:\ket{j}\ket{0}\rightarrow\ket{j}\ket{\Lambda_j}$ such that for every $j\in\intg_d$, after measuring the superposition state $\ket{\Lambda_j}$, with probability $\geq 1-\delta$ the first register $\lambda$ of the measurement outcome satisfies $|\lambda-v_j|$. The goal is to find an index $j$ such that $v_j\leq \min_{k\in\intg_d}v_k+\epsilon$.

\begin{proof}[Proof of Theorem~\ref{thm:min_finding_approx}]
Without loss of generality, we assume $\log d$ is a natural number. Let $v^*=\min_{k\in\intg_d}v_k$. For every $j\in\intg_d$, we let $\ket{\Lambda_j}=\sqrt{p_j^\epsilon}\ket{\Lambda^\epsilon_j}+\sqrt{1-p_j^\epsilon}\ket{\Lambda^{\epsilon^\perp}_j}$, where $\ket{\Lambda^\epsilon_j}$ is the superposition over numbers that are $\epsilon$-approximations of $v_j$, $\ket{\Lambda^{\epsilon^\perp}_j}$ is the  superposition over numbers that are not $\epsilon$-approximations of $v_j$, and $p^\epsilon_j \geq 1-\delta_2$ for every $j\in \intg_d$. Suppose we have a unitary $A$ that maps $\ket{j}\ket{0}\rightarrow\ket{j}\ket{\Lambda^\epsilon_j}$ and let $U=A(H^{\otimes \log d}\otimes I)$. Then one can see that if we apply the algorithm of Theorem~\ref{thm:generalized_min_finding} with the unitary $U$, then after using $K=1000\sqrt{d}\cdot\log(1/\delta_1)\geq1000/\sqrt{\Pr[X\leq v^*+\epsilon]}\cdot\log(1/\delta_1)$ applications of $A$ and $A^\dagger$, and $\mathcal{\tilde{O}}(\sqrt{d})$ elementary gates, with probability$\geq 1-\delta_1$, the first outcome $\lambda$ of the second register satisfies $\lambda\leq v^*+\epsilon$. Note that if $\lambda\leq v^*+\epsilon$, then the corresponding state $\ket{\Phi}$ satisfies that after measuring in the computational basis, the outcome $j$ satisfies $v_j\leq v^*+2\epsilon$. Therefore one can find a $j$ such that $v_j\leq v^*+2\epsilon$.

By the deferred measurement principle, one can consider the algorithm above as applying the unitary $\mathcal{A}=U_0E_0U_1E_2\cdots U_{K-1}E_{K-1}$ to the state $\ket{0}\ket{0}$ and measuring in the computational basis to get an outcome, where $U_i\in\{U,U^\dagger\}$ and $E_i$ is a circuit of elementary gates. Let us consider $\tilde{U}=\tilde{A}(H^{\otimes \log d}\otimes I)$, and let 
\[
\ket{\tilde{\psi}}=\tilde{U}\ket{0}\ket{0}=\frac{1}{\sqrt{d}}\sum\limits_{j\in\intg_d}\ket{j}\ket{\Lambda_j}\mbox{ and } \ket{{\psi}}=\frac{1}{\sqrt{d}}\sum\limits_{j\in\intg_d}\ket{j}\ket{\Lambda_j^\epsilon}=\alpha\ket{\tilde{\psi}}+\beta\ket{\tilde{\psi}^\perp}. 
\]
Where is a positive real, $\alpha \geq \sqrt{1-\delta_2}$ because $p^\epsilon_j\geq 1-\delta_2$ for every $j$, and $\beta=\sqrt{1-\alpha^2}$. 

\begin{claim}
There exists a unitary $U$ such that $U\ket{0}\ket{0}=\ket{\psi}$ and $\|U-\tilde{U}\|_{op}\leq \sqrt{2\delta_2}$.
\end{claim}
\begin{claimproof}
Define a unitary $V$ such that 
\begin{itemize}
    \item $V\ket{\tilde{\psi}}=\ket{\psi}$.
    \item $V\ket{\tilde{\psi}^\perp}=-\beta\ket{\tilde{\psi}}+\alpha\ket{\tilde{\psi}^\perp}$.
    \item For every $\ket{\phi}$ orthogonal to span$\{\ket{\psi},\ket{\tilde{\psi}}\}$, $V\ket{\phi}=\ket{\phi}$.
\end{itemize}
Let $U=V\tilde{U}$. One can see that $U\ket{0}\ket{0}=V\tilde{U}\ket{0}\ket{0}=V\ket{\tilde{\psi}}=\ket{\psi}$. Also, if we consider orthonormal basis $\{\ket{\tilde{\psi}},\ket{\tilde{\psi^\perp}},\ket{\phi_2},\ket{\phi_3},\ldots,\ket{\phi_{d-1}}\}$, then $V$ will be $\begin{pmatrix}
\alpha  & -\beta\\
\beta  & \alpha
\end{pmatrix} \oplus I_{d-2}$, and hence
\begin{align*}
    \|U-\tilde{U}\|_{op}&=\|I-V\|_{op}=\|\begin{pmatrix}
1-\alpha  & \beta\\
-\beta  & 1-\alpha
\end{pmatrix}\|_{op}\\
    &=\max\limits_{\substack{a,b\in \mathbb{C},\\ \text{s.t. } |a|^2+|b|^2=1}} \sqrt{|a(1-\alpha)+b\beta|^2+ |b(1-\alpha)-a\beta|^2}\\
    &= \sqrt{|1-\alpha|^2+|\beta|^2} \leq \sqrt{2\delta_2}.
\end{align*}
This proves the claim.
\end{claimproof}

\smallskip

Now let us consider the unitary $\tilde{\mathcal{A}}=\tilde{U}_0E_0\tilde{U}_1E_2\cdots \tilde{U}_{K-1}E_{K-1}$ applying to the state $\ket{0}\ket{0}$ and measuring in the computational basis to get an outcome, where $\tilde{U}_i\in\{\tilde{U},\tilde{U}^\dagger\}$. Because $\|\tilde{\mathcal{A}}-{\mathcal{A}}\|_{op}\leq K\cdot \|U-\tilde{U}\|_{op}\leq1000\log(1/\delta_1)\sqrt{2d\delta_2}$, with probability$\geq 1-\delta_1-1000\log(1/\delta_1)\sqrt{2d\delta_2}$, the first outcome $\lambda$ of the second register of $\tilde{\mathcal{A}}\ket{0}\ket{0}$ also satisfies $\lambda\leq v^*+\epsilon$, and hence we also find a $j$ such that $v_j\leq v^*+2\epsilon$ by the same arguments as above.
\end{proof}

\section{Classical lower bound for Lasso}\label{app:classicalLBforLasso}

In this section, we will give a classical lower bound for Lasso. Here we first introduce some tools we will use. The first one is the \emph{hypergeometric distribution} $Hyp(N,L,m)$ with parameters $N$, $L$, and $m$: the distribution of the number of marked balls drawn
when $m$ balls are drawn without replacement from a set of $n$ balls, $L$ of which are marked and $N-L$ are unmarked. On the other hand, the $\emph{binomial distribution}$ $Bin(m,L/N)$ is the distribution on the number of marked balls drawn with replacement. Holmes showed that when $m$ is small enough, the total variation distance between those two distribution will be very small.

\begin{theorem}[\cite{Holmes03}, Theorem~3.1]\label{thm:hyptobin}
Let $N,L,m \in \mathbb{N}$. If $N \geq L \geq m$, then
\[
d_{TV}( Hyp(N, L, m), Bin(m, L/N))\leq (m - 1)/(N - 1)
\]
\end{theorem}

We also use another distance between probability distributions.

\begin{definition}
Given two discrete probability distributions $\mathcal{P}$, $\mathcal{Q}$ over $\intg_N$, the Hellinger distance $d_H$ between $\mathcal{P}$ and $\mathcal{Q}$ is defined as
\[
d_H(\mathcal{P}, \mathcal{Q})\coloneqq \sqrt{\frac{1}{2}\sum\limits_{i
\in \intg_N} (\sqrt{\mathcal{P}_i}- \sqrt{\mathcal{Q}_i})^2}=\sqrt{1-\sum\limits_{i
\in \intg_N}\sqrt{\mathcal{P}_i\mathcal{Q}_i}}.
\]
\end{definition}

From the definition above we also have the following property for product distributions:
\begin{equation}\label{eq:Hellingerprod}
d^2_H(\mathcal{P}^{\otimes m}, \mathcal{Q}^{\otimes m})=1-(\sum\limits_{i
\in \intg_N}\sqrt{\mathcal{P}_i\mathcal{Q}_i})^m=1-(1-d^2_H(\mathcal{P}, \mathcal{Q}))^m\leq m \cdot d^2_H(\mathcal{P}, \mathcal{Q}).
\end{equation}
The following lemma bridges the Hellinger distance and total variation distance.

\begin{lemma}\label{lem:dH_to_dTV}
For arbitrary discrete probability distributions $\mathcal{P}$, $\mathcal{Q}$ over $\intg_N$, we have
\[
d^2_H(\mathcal{P}, \mathcal{Q}) \leq d_{TV}(\mathcal{P}, \mathcal{Q}) \leq \sqrt{2}\,d_H(\mathcal{P}, \mathcal{Q}).
\]
\end{lemma}

\begin{proof}
First we prove the left inequality:
\begin{align*}
d^2_H(\mathcal{P}, \mathcal{Q}) &=\frac{1}{2}\sum\limits_{i
\in \intg_N} (\sqrt{\mathcal{P}_i}- \sqrt{\mathcal{Q}_i})^2\leq \frac{1}{2}\sum\limits_{i
\in \intg_N} (|\sqrt{\mathcal{P}_i}- \sqrt{\mathcal{Q}_i}|)(\sqrt{\mathcal{P}_i}+ \sqrt{\mathcal{Q}_i}) \\
&= \frac{1}{2}\sum\limits_{i
\in \intg_N} | \mathcal{P}_i-\mathcal{Q}_i | = d_{TV}(\mathcal{P}, \mathcal{Q}). 
\end{align*}
The right inequality follows using Cauchy-Schwarz:
\begin{align*}
    d_{TV}(\mathcal{P}, \mathcal{Q}) &=\frac{1}{2}\sum\limits_{i
\in \intg_N} |\sqrt{\mathcal{P}_i}- \sqrt{\mathcal{Q}_i}|(\sqrt{\mathcal{P}_i}+ \sqrt{\mathcal{Q}_i}) \\
&\leq \frac{1}{2}\sqrt{\sum\limits_{i\in\intg_N}|\sqrt{\mathcal{P}_i}- \sqrt{\mathcal{Q}_i}|^2}\sqrt{\sum\limits_{i\in\intg_N}|\sqrt{\mathcal{P}_i}+ \sqrt{\mathcal{Q}_i}|^2}\\
&=d_H(\mathcal{P}, \mathcal{Q})\cdot \sqrt{1+\sum\limits_{i\in\intg_N}\sqrt{\mathcal{P}_i\mathcal{Q}_i}}=d_H(\mathcal{P}, \mathcal{Q})\cdot\sqrt{2-d^2_H(\mathcal{P}, \mathcal{Q})}\leq\sqrt{2}\,d_H(\mathcal{P}, \mathcal{Q}).
\end{align*}
\end{proof}

After introducing the above tools, we are ready to prove the following result.

\begin{theorem}\label{thm:Ham_distinguisher}
Let $m\in \intg_+$, $p\in (0,0.5)$, and $\mathcal{P},\mathcal{Q}$ be two hypergeometric distributions $Hyp(N,N/2,m)$ and $Hyp(N,N/2+pN,m)$ respectively. Then we have $d_{TV}(\mathcal{P},\mathcal{Q}) \leq 2(m-1)/(N-1)+ p\sqrt{3m}$. %If $m \leq 1/(100p^2)\leq N/100$, then we have $ d_{TV}(\mathcal{P},\mathcal{Q}) \leq 1/3$.
\end{theorem}

\begin{proof}
Let $\mathcal{P'},\mathcal{Q'}$ be the binomial distributions $Bin(m,1/2)$ and $Bin(m,1/2+p)$ respectively. By the triangle inequality, Theorem~\ref{thm:hyptobin}, and Lemma~\ref{lem:dH_to_dTV}, we have
\begin{align*}
    d_{TV}(\mathcal{P},\mathcal{Q}) &\leq d_{TV}(\mathcal{P'},\mathcal{P})+ d_{TV}(\mathcal{P'},\mathcal{Q'})+ d_{TV}(\mathcal{Q'},\mathcal{Q}) \\
    &\leq (m-1)/(N-1) + d_{TV}(\mathcal{P'},\mathcal{Q'}) + (m-1)/(N-1)\\
    &\leq 2(m-1)/(N-1)+\sqrt{2}\,d_{H}(\mathcal{P'},\mathcal{Q'}).
\end{align*}
Suppose $\textbf{p}'$ and $\textbf{q}'$ are Bernoulli distributions with mean $1/2$ and $1/2+p$ respectively, then using Eq.~\eqref{eq:Hellingerprod}
%the fact that  %$d_H(\textbf{p}'^{\otimes m}, \textbf{q}'^{\otimes m})\leq \sqrt{m}\cdot d_H(\textbf{p}', \textbf{q}')$ 
we have
\begin{align*}
    d_{H}(\mathcal{P'},\mathcal{Q'})\leq d_H(\textbf{p}'^{\otimes m}, \textbf{q}'^{\otimes m}) \leq \sqrt{m}\cdot d_H(\textbf{p}', \textbf{q}') \leq \sqrt{m}\cdot \sqrt{1-\frac{1}{2}(\sqrt{1+2p}+\sqrt{1-2p})}\leq p\sqrt{\frac{3m}{2}},
\end{align*}
where the last inequality holds because $\sqrt{1+2p}+\sqrt{1-2p}\geq 2-3p^2$ for $p\in (0,0.5)$. 
Combining the above two results, we have
\[
d_{TV}(\mathcal{P},\mathcal{Q}) \leq 2(m-1)/(N-1)+\sqrt{2}d_{H}(\mathcal{P'},\mathcal{Q'}) \leq 2(m-1)/(N-1)+ p\sqrt{3m}. %\leq 1/3,
\]
%where the last inequality holds because $m \leq 1/(100p^2)\leq N/100$.
\end{proof}

%The above theorem gives a classical query lower bound for the Hamming-weight distinguisher problem $\text{HD}_{\frac{N}{2},N(\frac{1}{2}+p)}$ defined in Section~\ref{sec:q_lower_decomposed}.

%\begin{corollary}\label{cor:Ham_dist_classic}
%Let $N\in 2\intg_+$, $z\in \{0,1\}^N$, and $p\in (1/\sqrt{N},0.5)$ be a multiple of $1/N$. Suppose we have query access to $z$.
%Then every bounded-error classical algorithm that computes $\mbox{HD}_{\frac{N}{2},N(\frac{1}{2}+p)}$ makes $\Omega(1/p^2)$ queries.
%\end{corollary}

%Theorem~\ref{thm:exact-set-finding} tells us that if there is an algorithm that can solve $\mbox{WSF}_{d,w,p,N}$, then by only few many extra queries and applications of the algorithm, we can find the hidden set $W$ exactly.
Using the above theorem, we can show a classical query lower bound for the \emph{exact set-finding problem} $\mbox{ESF}_{d,w,p,N}$, which is the following: given a matrix $X\in\{-1,1\}^{N\times d}$ where each column-sum is either $2pN$ or $0$, the goal is to find the set $W$ (of size $w$ or $w-1$) of indices of the columns whose entries add up to $2pN$. The following theorem gives a classical query lower bound for this problem.

\begin{theorem}\label{thm:ESF_LB}
Let $N\in 2\intg_+$, $p\in (1/\sqrt{N},0.5)$ be an integer multiple of $1/N$, and $w=1/(2p)$. Given a matrix $X \in \{-1,1\}^{N\times d}$ such that there exists a set $W\subset \intg_d$ with size $w$ or $w-1$ satisfying \begin{itemize}
    \item for every $j \in W$, $\sum\limits_{i\in\intg_N} X_{ij}=2pN$;
    \item for every $j' \in \overline{W}$, $\sum\limits_{i\in\intg_N} X_{ij'}= 0$.
\end{itemize}
Suppose we have classical query access to $X$. Then every classical algorithm that computes $W$ with success probability $\geq 1-1/100$ uses $\Omega((d-w)/p^2)$ queries.
\end{theorem}

\begin{proof}
Let $W=\{1,\ldots,w-1\}$ and let $\mathcal{D}_W$ be the distribution on the input~$X$ where each column of $X$ is chosen uniformly at random subject to the column sums as specified in the theorem: if $j\in W$ then the $j$th column-sum is $2pN$, and if $j\not \in W$ then the $j$th column-sum is~0. %Let 
%\[
%\mathcal{D}=\frac{1}{\binom{d}{w-1}}\sum_{W\subseteq \intg_d,|W|=w-1}\mathcal{D}_W, 
%\]
%which is uniform over all matrices $X$ that have $w-1$ biased columns.

%For every fixed $W_\uparrow$, we also define the distribution $\mathcal{D}^{(\frac{1}{2})}_{W_\uparrow}$ on the input~$X$ in the following. We toss a fair coin; if it is heads, then we choose another index $k$ from $\intg_d\setminus W_\uparrow$ uniformly at random and let $W=W_\uparrow \cup \{k\}$ (then $|W|=w$); and if it is tails, then we will let $W=W_\uparrow$ (then $|W|=w-1$). Once $W$ has been chosen, we pick each column of $X$ uniformly at random subject to the column sums as specified in the theorem. Note that we have $\mathcal{D}^{(\frac{1}{2})}=\frac{1}{\binom{d}{w-1}}\sum_{W_\uparrow \subseteq \intg_d} \mathcal{D}^{(\frac{1}{2})}_{W_{\uparrow}}$. Given $W_\uparrow$, we can also define the distribution $\mathcal{D}_{W_{\uparrow}}$ on the input~$X$ in the following\rnote{you already defined this distribution in the first paragraph, you just didn't name it there; try to avoid such repetition}. We let $W=W_{\uparrow}$, and then pick each column of $X$ uniformly at random subject to the column sums as specified in the theorem.

Let $\mathcal{A}$ be a randomized classical algorithm with worst-case query complexity~$T$ that computes the hidden set with error probability $\leq 1/100$ for  every input. %Then by averaging\ynote{do we really need this?} there exists a specific set $W\subseteq \mathbb{Z}_d$ of size $w-1$ such that the error probability $\delta$ of algorithm $\mathcal{A}$ is at most $1/100$ under distribution~$\mathcal{D}_W$. 
Let random variable $t_j$ be the number of queries that algorithm~$\mathcal{A}$ makes in the $j$th $N$-bit string (i.e., to entries in the $j$th column of $X$) under~$\mathcal{D}_W$, and $T_j$ be the expectation of $t_j$ under~$\mathcal{D}_W$. Because $\sum_{j\in\intg_d}T_j\leq T$, there 
must be a $k\in \intg_d\setminus{W}$ such that $T_{k}\leq T/(d-w+1)$. We use algorithm $\mathcal{A}$ to prove the following claim:

\bigskip

\begin{claim}
There exists a classical randomized algorithm with worst-case query complexity $\leq 100T/(d-w+1):=\tau$ in the $k$-th column that distinguishes $\mathcal{D}_W$ from $\mathcal{D}_{W\cup \{k\}}$ with success probability $\geq 0.98$.
\end{claim}

\medskip

\begin{claimproof}
Let $\mathcal{B}$ be the following randomized classical algorithm: run $\mathcal{A}$ until the number of queries in $k$-th column is $\geq \tau$. If $\mathcal{A}$ outputs $W$, then we output that; if $\mathcal{A}$ does not output $W$ or did not terminate within $\tau$ queries, then we output $W\cup \{k\}$.

We here prove the correctness of algorithm~$\mathcal{B}$. If we run $\cal B$ on input distribution $\mathcal{D}_W$, then the probability that $\cal A$ (run all the way until it terminates) does not output~$W$, is $\leq 1/100$. By Markov's inequality, the probability (still under distribution $\mathcal{D}_W$) that $\cal A$ did not terminate within $\tau$ queries, $\leq 1/100$. Hence by the union bound, the probability that $\cal B$ does not output $W$ is $\leq 2/100$. 

If, on the other hand, we run $\cal B$ on input distribution is $\mathcal{D}_{W\cup \{k\}}$, then the probability that $\cal B$ outputs the correct set $W\cup\{k\}$ is lower bounded by the probability that $\cal A$ outputs  $W\cup\{k\}$, because we defined $\cal B$ to output $W\cup\{k\}$ when $\cal A$ did not already terminate. 
Since $\cal A$ has success probability $\geq 99/100$, the probability (under $\mathcal{D}_{W\cup \{k\}}$) that $\cal B$ outputs $W\cup\{k\}$ is  $\geq 99/100$.
\end{claimproof}

\bigskip

%Let $\delta_{k}$ be the error probability of algorithm $\mathcal{A}$ under the distribution $\mathcal{D}_k=\frac{1}{2}\mathcal{D}_W+\frac{1}{2}\mathcal{D}_{W\cup \{k\}}$.% We know $\delta_{k}/d\leq \delta$, implying $\delta_{k}\leq 1/100$.

%We now convert algorithm $\mathcal{A}$ into an algorithm $\mathcal{A'}$ with worst-case running time 100T and the error probability $\leq 1/100$ by running the original algorithm for $100T$ steps and cutting it off if it does not stop.
Because algorithm $\mathcal{B}$ decides with success probability $\geq 98/100$ whether $k$ is in the hidden set or not, it has to distinguish (in the $k$th column of $X$) a uniformly random column with column-sum 0 from a uniformly random column with column-sum $2pN$ with success probability $\geq 98/100$. Hence we must have
\begin{align*}
    \Omega(1)  %\E_{\mathcal{D}_{W}}[d_{TV}(Hyp(N,N/2,t_{k}),Hyp(N,N/2+pN,t_{k}))] \\
 %   &\leq \E_{\mathcal{D}_{W}}\left[\frac{2(t_{k}-1)}{N-1}+p\sqrt{3t_{k}}\right]\\
    \leq \frac{2(100T_{k}-1)}{N-1}+p\sqrt{300T_k}
    \leq \frac{200T}{(d-w+1)(N-1)}+p\sqrt{\frac{300T}{(d-w+1)}},
\end{align*}
where %the expectation is taken over the distribution $\mathcal{D}_{W}$, which determines the actual number of queries~$t_{k}$.
the first inequality follows by Theorem~\ref{thm:Ham_distinguisher} % , the third one by Jensen's inequality, 
 and the last one follows because %$\E_{\mathcal{D}_{W}}[t_{k}]=T_{k}$ and 
$T_{k}\leq T/(d-w+1)$.
%Hence by fixing all $d-1$ other $N$-bit strings, we have a classical algorithm with $T_k$ queries. 
%By Corollary~\ref{cor:Ham_dist_classic}, we know $T_k=\Omega(1/p^2)$, 
Rearranging implies $T =\Omega((d-w)/p^2)$.
\end{proof}

The last step towards our lower bound for classical Lasso-solvers is to show that one can solve the exact set-finding problem using an approximate Lasso-solver, as follows.

\begin{algorithm}[hbt]
\SetKwData{Left}{left}\SetKwData{This}{this}\SetKwData{Up}{up}
\SetKwFunction{Union}{Union}\SetKwFunction{FindCompress}{FindCompress}
\SetKwInOut{Input}{input}\SetKwInOut{Output}{output}
\Input{Algorithm $\mathcal{A}$ that outputs (with probability $\geq 1-1/(20000\log d)$) a $p/4000$-minimizer  for Lasso with respect to $\mathcal{D}_{p,W}$ for every possible $W\subset \intg_d$ of size $w$ or $w-1$, using $M$ samples; $X\in \{-1,1\}^{N\times d}$ a valid input for $\mbox{ESF}_{d,w,p,N}$;} 
\SetAlgoLined
      \For{$u=0$ to $U-1=\lceil 100\log d\rceil-1$}{
      GENERATE a permutation $\pi\in S_d$ uniformly at random\;
      GENERATE $R\in \intg_N^{M\times d}$ such that all its entries are i.i.d.\ samples from $\cU_N$\;
      Let $X'\in\{-1,1\}^{M \times d}$ be $X_{ij}'=X_{R_{i\pi(j)}\pi(j)}$; let $X'_m$ denote its $m$th row\;
  RUN Algorithm~$\mathcal{A}$ with inputs $(X'_1,1), (X'_2,1),\ldots, (X'_M,1)$ and let $W'_u\subseteq\intg_d$ be the set of indices of entries of the output of Algorithm~$\mathcal{A}$ whose absolute value is $\geq 2p/3$\;
  STORE $W_u=\pi^{-1}(W'_u)$\;
  }
  \Output {$\tilde{W}=\{j\in \intg_d : j \mbox{ is included in at least half of the sets } W_0,\ldots, W_{U-1}\}$;}
 \caption{Solve $\mbox{ESF}_{d,w,p,N}$ using a $p/4000$-Lasso solver $\mathcal{A}$}
 \label{Alg:ESFtoLasso}
\end{algorithm}

\begin{theorem}
Let $N\in 2\intg_+$, $p\in (1/\sqrt{N},0.25)$ be an integer multiple of $1/N$, and $w=1/(2p)$. Suppose $\mathcal{A}$ is an algorithm that finds a $p/4000$-minimizer for Lasso with respect to $\mathcal{D}_{p,W}$ with probability $\geq 1-1/(20000\log d)$ for every $W\subset \intg_d$ with size $w$ or $w-1$. Then Algorithm~\ref{Alg:ESFtoLasso} outputs the correct answer for $\mbox{ESF}_{d,w,p,N}$ with success probability $\geq 99/100$. 
\end{theorem}

\begin{proof} Let $u\in\intg_U$ and define $p_{j}=\Pr[j\mbox{ is in  }W_u]$.
Note that $p_j$ is independent of~$u$ because all iterations do the same thing. 
First we show that $\forall u\in\intg_U$ and $\forall j_1,j_2\in W$, $p_{j_1}=p_{j_2}$. Let $p_{j,\pi}=\Pr[j\in \pi^{-1}(W'_u)\mid \pi]$ be the probability of the event that index $j$ is in $\pi^{-1}(W'_u)=W_u$ if the $u$th iteration used the permutation $\pi$. We can see that $p_{j}=\frac{1}{d!}\sum\limits_{\pi\in S_d} p_{j,\pi}$. Consider a permutation $\sigma\in S_d$ satisfying that
\[
\sigma(j)=
    \begin{dcases}
    j_2, & \text{if } j=j_1,\\
    j_1, & \text{if } j=j_2,\\
    j, &\text{otherwise}.
\end{dcases}
\]
We have for every $\pi\in S_d$, that $p_{j_1,\pi}=p_{j_2,\sigma\pi}$ because both $j_1$ and $j_2$ are in $W$ and each of them is drawn from the corresponding columns with replacement. Therefore, we have 
\[
p_{j_1}=\frac{1}{d!}\sum\limits_{\pi\in S_d} p_{j_1,\pi}= \frac{1}{d!}\sum\limits_{\pi\in S_d} p_{j_2,\sigma\pi}=p_{j_2},
\]
and using a similar argument, we can also show that for arbitrary $j_1',j_2'\in \overline{W}$, we have $p_{j'_1}=p_{j'_2}$.

Combining the above argument and Theorem~\ref{thm:average_HSF_to_Lasso},
%which tells us that the index set that collects all indices of entries of the output of Algorithm~$\mathcal{A}$ whose absolute value is $\geq 2p/3$ contains at least $99.5\%$ elements of $W$, 
we have that if algorithm~$\mathcal{A}$ succeeds, then for every $j\in W$, $p_j\geq 0.995$ and for every $j'\in \overline{W}$, $p_{j'}\leq 0.005$. This implies that for every $j\in W$, $\Pr[j\notin \tilde{W}]\leq 1/200d$ and for every $j' \in \intg_d\setminus W$, $\Pr[j'\in \tilde{W}]\leq 1/200d$ by Hoeffding bound. Also because algorithm~$\mathcal{A}$ outputs a $p/4000$-minimizer with error probability at most $1/(20000\log d)$, by union bound the probability that $\cal A$ fails in at least one of the $U$ inner loops is at most $1/200$. Hence $\tilde{W}$ is the correct answer for $\mbox{ESF}_{d,w,p,N}$ with probability $\geq 1-1/200-d/200d= 99/100$.
\end{proof}

Similar to the arguments in Section~\ref{sec:worst_to_average}, the above theorem tells us that we can convert an instance of $\mbox{ESF}_{d,w,p,N}$ to an instance of $\mbox{DSF}_{\mathcal{D}_{p,W}}$ (and hence an instance for Lasso). Again the matrix $R$ is produced offline and therefore we do not use extra queries in Algorithm~\ref{Alg:ESFtoLasso} apart from the runs of~$\cal A$. Now suppose we have a $T$-query classical algorithm that outputs (with success probability $\geq 1-1/(20000\log d)$) a $p/4000$-minimizer for Lasso with respect to $\mathcal{D}_{p,W}$ for arbitrary $W\subset \intg_d$ of size $w$ or $w-1$, then using $T\cdot 100 \log d$ queries, we can solve $\mbox{ESF}_{d,w,p,N}$ with probability $\geq 99/100$. %(and we can increase the success probability of solving $\mbox{ESF}_{d,w,p,N}$\rnote{which algorithm? for what?} from $2/3$ to $1-1/(100d)$ by repeating this $2/3$-success-probability algorithm of solving $\mbox{ESF}_{d,w,p,N}$ $\mathcal{O}(\log d)$ times and taking the majority of their output)
Note that we can construct such a high-success-probability Lasso solver by applying a Lasso solver with success probability $\geq2/3$ for $\mathcal{O}(\log\log d)$ times and then outputting the output vector with the smallest objective value (estimating objective values with additive error $p$ with success probability $\geq 1-\delta$ uses only $\mathcal{\tilde{O}}(\log (1/\delta)/p^2)$ queries). %\ynote{is this really easy to be done? If it is not, then maybe we should modify the success probability of Corollary B.7}
Also, Theorem~\ref{thm:ESF_LB} gives a $\Omega((d-w)/p^2)$ lower bound for $\mbox{ESF}_{d,w,p,N}$, and hence we obtain a classical lower bound of $\tilde{\Omega}(d/\epsilon^2)$ queries for Lasso for $\epsilon \in (1/\sqrt{d},1/200)$ by letting $p=\epsilon/2$ and by the fact $\lfloor 1/\epsilon \rfloor\geq w \geq\lfloor 1/\epsilon \rfloor-1$:

\begin{corollary}
Let $\epsilon\in (1/\sqrt{d}, 1/200)$, $w$ be either $\lfloor 1/\epsilon\rfloor$ or $\lfloor 1/\epsilon\rfloor-1$, $p=1/(2\lfloor 1/\epsilon\rfloor)$, and $W\subset \intg_d$ with size $w$. Every bounded-error classical algorithm that computes an $\epsilon$-minimizer for Lasso with respect to $\mathcal{D}_{p,W}$ uses $\tilde{\Omega}(d/\epsilon^{2})$ queries.
\end{corollary}

\end{document}